\documentclass{eptcs}
\usepackage{breakurl}
\usepackage{underscore}
\usepackage{amsmath,amsfonts,amsthm,amssymb,stmaryrd,relsize}

\theoremstyle{plain}

\numberwithin{equation}{section}
\newtheorem{thm}{Theorem}[section]
\newtheorem{lem}[thm]{Lemma}
\newtheorem{cor}[thm]{Corollary}

\allowdisplaybreaks  

\newcounter{cond}

\newcommand{\complex}{{\mathbb C}}
\newcommand{\real}{{\mathbb R}}

\newcommand{\ascript}{{\mathcal A}}
\newcommand{\bscript}{{\mathcal B}}
\newcommand{\cscript}{{\mathcal C}}
\newcommand{\dscript}{{\mathcal D}}
\newcommand{\escript}{{\mathcal E}}
\newcommand{\fscript}{{\mathcal F}}
\newcommand{\gscript}{{\mathcal G}}
\newcommand{\hscript}{{\mathcal H}}
\newcommand{\sscript}{{\mathcal S}}

\newcommand{\ahat}{\widehat{a}}
\newcommand{\bhat}{\widehat{b}}
\newcommand{\capahat}{\widehat{A}}
\newcommand{\capbhat}{\widehat{B}}
\newcommand{\chat}{\widehat{c}}
\newcommand{\alphabar}{\overline{\alpha}}
\newcommand{\ascripthat}{\widehat{\ascript}}

\newcommand{\Omegahat}{\widehat{\Omega}}
\newcommand{\atilde}{\widetilde{a}}
\newcommand{\btilde}{\widetilde{b}}
\newcommand{\ctilde}{\widetilde{c}}
\newcommand{\caputilde}{\widetilde{U}}

\newcommand{\ctimes}{\mathrel{\mathlarger\cdot}}

\newcommand{\ab}[1]{\left|#1\right|}
\newcommand{\doubleab}[1]{\left|\left|#1\right|\right|}
\newcommand{\brac}[1]{\left\{#1\right\}}
\newcommand{\paren}[1]{\left(#1\right)}
\newcommand{\sqbrac}[1]{\left[#1\right]}
\newcommand{\sqparen}[1]{{\left[#1\right)}}
\newcommand{\elbows}[1]{{\left\langle#1\right\rangle}}
\newcommand{\ceils}[1]{{\left\lceil#1\right\rceil}}

\title{Contexts in Convex and Sequential Effect Algebras}
\author{Stan Gudder
  \institute{Department of Mathematics\\
    University of Denver\\ Denver, Colorado 80208}
  \email{sgudder@du.edu}
}
\def\titlerunning{Contexts in Convex and Sequential Effect Algebras}
\def\authorrunning{S. Gudder}

\begin{document}
\maketitle

\begin{abstract}
A convex sequential effect algebra (COSEA) is an algebraic system with three physically motivated operations, an orthogonal sum, a scalar product and a sequential product. The elements of a COSEA correspond to yes-no measurements and are called effects. In this work we stress the importance of contexts in a COSEA. A context is a finest sharp measurement and an effect will act differently according to the underlying context with which it is measured. Under a change of context, the possible values of an effect do not change but the way these values are obtained may be different. In this paper we discuss direct sums and the center of a COSEA. We also consider conditional probabilities and the spectra of effects. Finally, we characterize COSEA's that are isomorphic to COSEA's of positive operators on a complex Hilbert space. These result in the traditional quantum formalism. All of this work depends heavily on the concept of a context.

\end{abstract}


\section{Introduction}  
We present an axiomatic framework for quantum mechanics in which the basic entities and operations have physical significance. In this framework, the principle concepts are states and effects. The states represent initial preparations that describe the condition of the system, while the effects represent yes--no measurements that probe the system. The effects may be unsharp or fuzzy \cite{dp94,fb94,gud982,kra83}. A state applied to an effect produces the probability that the effect gives a yes value when the system is in that state. The resulting mathematical structure is called a convex sequential effect algebra (COSEA) $\escript$ \cite{gud99,gp98,gud18,wet181,wet182}. The three mathematical operations in $\escript$ are an orthogonal sum $a\oplus b$, a scalar product $\lambda a,\lambda\in [0,1]\subseteq\real$ and a sequential product $a\circ b$. These operations have physical interpretations that we now discuss.

Although this framework is much more general, we can employ the model of an optical bench to visualize what is happening here. A beam of particles (photons, electrons, etc.) is emitted from a source and propagates through a channel on the bench until the beam arrives at a detector at the end of the channel. The particles are initially prepared in a certain state and the effects describe various filters that can be placed in the channel. The beam travels through one or more filters which interact with the beam and can change its properties in certain ways. The detector may count particles or measure different characteristics of the beam. The sum $a\oplus b$ is performed by first splitting the beam into two equal parts, which are directed toward the two filters placed in parallel after which both beams are reunited before being collected at the detector. The scalar product $\lambda a$ corresponds to an attenuation of filter $a$ by the factor $\lambda$. This can be accomplished by placing a gray filter with a certain darkness in front of filter $a$. The gray filter blocks some of the particles but does not otherwise disturb the beam. The sequential product $a\circ b$ is performed by placing the filters in series so that $a$ is first and $b$ is second. In this way, filter $a$ can interfere with the operation of filter
$b$ while $b$ cannot interfere with the operation of $a$. We will find this useful for describing quantum interference.

In this work, an important role will be played by the context under which an effect is observed. A context is a finest sharp measurement and an effect will act differently according to the underlying context with which it is measured. For example, in the optical bench scenario, changing contexts may result from altering the detectors or varying the size, shape or location of the bench. Under a change of context, the possible values of an effect do not change but the way these values are obtained may be different. As far as contexts are concerned, there is a great difference between classical and quantum systems. We shall show that classical systems have exactly one context, while quantum systems have infinitely many.

In Section 2 we define the concepts of COSEA's and contexts. Section~3 discusses direct sums and the center of a COSEA. Section~4 considers conditional probabilities and spectra of effects. Finally, Section~5 characterizes COSEA's that are isomorphic to COSEA's of positive operators on a complex Hilbert space. Of course, these result in the traditional quantum formalism. There is some overlap of this paper and the work in \cite{wet181,wet182}. However, our stress on contexts provides a different approach.

\section{Convex Sequential Effect Algebras} 
Let $\escript$ be the set of effects and $\sscript$ the set of states for a physical system. The connection between $\escript$ and $\sscript$ is given by a \textit{probability function} $F\colon\escript\times\sscript\to [0,1]\subseteq\real$ where $F(a,s)$ is interpreted as the probability that effect $a$ has a yes value when the system is in state $s$. An \textit{effect-state space} is a triple $(\escript ,\sscript ,F)$ where $\escript$ and $\sscript$ are nonempty sets and $F\colon\escript\times\sscript\to [0,1]$ satisfies:
\begin{list} {(ES\arabic{cond})}{\usecounter{cond}
\setlength{\rightmargin}{\leftmargin}}
\item There exist elements $0,1\in\escript$ such that $F(0,s)=0$, $F(1,s)=1$ for every $s\in\sscript$.
\item If $F(a,s)\le F(b,s)$ for every $s\in\sscript$, then there exists a unique $c\in\escript$ such that $F(a,s)+F(c,s)=F(b,s)$ for all
$s\in\sscript$.
\item If $a\in\escript$ and $\lambda\in\sqbrac{0,1}$, then there exists an element $\lambda a\in\escript$ such that
$F(\lambda a,s)=\lambda F(a,s)$ for all $s\in\sscript$.
\end{list}

The elements $0,1$ in (ES1) correspond to the null effect that never occurs and the unit effect that always occurs, respectively. It is shown in
\cite{gud99,gp98} that if $F(a,s)+F(b,s)\le 1$ for every $s\in\sscript$, then there exists a unique $c\in\escript$ such that
\begin{equation*}
F(c,s)=F(a,s)+F(b,s)
\end{equation*}
for all $s\in\sscript$. We then write $a\perp b$ and define $a\oplus b=c$. In this way, $\oplus$ is a partial binary operation on $\escript$.

The structure $(\escript ,0,1,\oplus )$ is called an \textit{effect algebra} and satisfies the following axioms:
\begin{list} {(EA\arabic{cond})}{\usecounter{cond}
\setlength{\rightmargin}{\leftmargin}}
\item If $a\perp b$, then $b\perp a$ and $b\oplus a=a\oplus b$,
\item If $a\perp b$ and $(a\oplus b)\perp c$, then $b\perp c$, $a\perp (b\oplus c)$ and $a\oplus (b\oplus c)=(a\oplus b)\oplus c$,
\item For every $a\in\escript$ there exists a unique $a'\in\escript$ such that $a\perp a'$ and $a \oplus a'=1$,
\item If $a\perp 1$, then $a=0$.
\end{list}

We define $a\le b$ if there is a $c\in\escript$ such that $a\oplus c=b$. The element $c$ is unique and we write $c=b\ominus a$. It can be shown that $(\escript ,0,1,\le )$ is a bounded poset and $a\perp b$ if and only if $a\le b'$ \cite{dp94,fb94}. Moreover, $a''=a$ and $a\le b$ implies $b'\le a'$ for all $a,b\in\escript$. If we incorporate the scalar product $\lambda a$ of (ES3) we obtain the following structure. An effect algebra $\escript$ is
\textit{convex} \cite{gud99,gp98,gud18} if for every $a\in\escript$ and $\lambda\in [0,1]\subseteq\real$ there exists an element $\lambda a\in\escript$ such that
\begin{list} {(CO\arabic{cond})}{\usecounter{cond}
\setlength{\rightmargin}{\leftmargin}}
\item If $\alpha ,\beta\in\sqbrac{0,1}$ and $a\in\escript$, then $\alpha (\beta a)=(\alpha\beta )a$.
\item If $\alpha ,\beta\in\sqbrac{0,1}$ with $\alpha +\beta\le 1$ and $a\in\escript$, then $\alpha a\perp\beta b$ and
$(\alpha +\beta )a=\alpha a\oplus\beta a$.
\item If $a,b\in\escript$ with $a\perp b$ and $\lambda\in\sqbrac{0,1}$, then $\lambda a\perp\lambda b$ and
$\lambda (a\oplus b)=\lambda a\oplus\lambda b$.
\item If $a\in\escript$, then $1a=a$.
\end{list}

We call an effect algebra an EA and a convex effect algebra a COEA, for short. In $\escript$ and $\fscript$ are EA's, a map
$\phi\colon\escript\to\fscript$ is \textit{additive} if $a\perp b$ implies that $\phi (a)\perp\phi (b)$ and
\begin{equation*}
\phi (a\oplus b)=\phi (a)\oplus\phi (b)
\end{equation*}
An additive map $\phi$ that satisfies $\phi (1)=1$ is called a \textit{morphism}. A morphism $\phi\colon\escript\to\fscript$ for which
$\phi (a)\perp\phi (b)$ implies $a\perp b$ is a \textit{monomorphism} and a surjective monomorphism is an \textit{isomorphism}. If $\escript$ and $\fscript$ are COEA's, a morphism $\phi\colon\escript\to\fscript$ is \textit{affine} if $\phi (\lambda a)=\lambda\phi (a)$ for all
$\lambda\in [0,1]$, $a\in\escript$. If there exists an affine isomorphism $\phi\colon\escript\to\fscript$ we say that $\escript$ and $\fscript$ are COEA \textit{isomorphic}.

The simplest example of a COEA is the unit interval $[0,1]\subseteq\real$ with the usual addition (when $a+b\le 1$) and scalar multiplication.
A \textit{state} on an EA $\escript$ is a morphism $\omega\colon\escript\to [0,1]$. Notice that in an effect-state space, the function 
$a\mapsto F(a,s)$ is a state on $\escript$. We denote the set of states on $\escript$ by $\Omega (\escript )$. We say that
$S\subseteq\Omega (\escript )$ is \textit{order-determining} if $\omega (a)\le\omega (b)$ for all $\omega\in S$ implies that $a\le b$. It is shown in \cite{gp98} that every state on a COEA is affine. It is also shown in \cite{gp98} that an effect-state space is equivalent to a COEA with an order-determining set of states.

We now introduce the sequential product $a\circ b$ on a COEA. Because of the series order for $a\circ b$, $a$ may interfere with the $b$ measurement but $b$ will never interfere with the $a$ measurement. If $a\circ b=b\circ a$ we write $a\mid b$ and say that $a$ and $b$ \textit{do not interfere}. We now present our general definition.

A \textit{convex sequential effect algebra} (COSEA) \cite{gud18} is an algebraic system
$(\escript ,0,1,\oplus ,\circ )$ where\linebreak 
$(\escript ,0,1,\oplus )$ is a COEA and $\circ\colon\escript\times\escript\to\escript$ is a binary operation satisfying:
\begin{list} {(S\arabic{cond})}{\usecounter{cond}
\setlength{\rightmargin}{\leftmargin}}
\item $b\mapsto a\circ b$ is additive for all $a\in\escript$,
\item $1\circ a=a$ for all $a\in\escript$,
\item If $a\circ b=0$, then $a\mid b$,
\item If $a\mid b$, then $a\mid b'$ and $a\circ (b\circ c)=(a\circ b)\circ c$ for all $c\in\escript$,
\item If $c\mid a$ and $c\mid b$ then $c\mid a\circ b$ and $c\mid (a\oplus b)$ whenever $a\perp b$,
\item For all $\lambda\in [0,1]\subseteq\real$, $a,b\in\escript$, we have that
\begin{equation*}
(\lambda a)\circ b=a\circ (\lambda b)=\lambda (a\circ b)
\end{equation*}
\end{list}
It is shown in \cite{wet181} that if $\escript$ satisfies an additional continuity property that makes $\escript$ a $\sigma$-COSEA then (S6) is automatically satisfied.

In quantum mechanics, $a\circ b$ is useful for describing quantum interference. It is also needed for defining the important concept of conditional probability. An element $a$ in a COSEA is \textit{sharp} if the greatest lower bound $a\wedge a'=0$. Sharp effects are thought of as effects that are precise or unfuzzy. We denote the set of sharp effects in $\escript$ by $S(\escript )$.

\begin{thm}    
\label{thm21}
{\rm \cite{gg02}}\ The sequential product in a COSEA $\escript$ has the following properties.
{\rm (i)}\enspace $a\circ b\le a$ for all $a,b\in\escript$.
{\rm (ii)}\enspace If $a\le b$, then $c\circ a\le c\circ b$ for all $c\in\escript$.
{\rm (iii)}\enspace $a\in S(\escript )$ if and only if $a\circ a=a$.
{\rm (iv)}\enspace For $a\in\escript$, $b\in S(\escript )$, $a\circ b=0$ if and only if $a\perp b$.
{\rm (v)}\enspace For $a\in\escript$, $b\in S(\escript )$, $a\le b$ if and only if $a\circ b=b\circ a=a$ and $b\le a$ if and and only if
$a\circ b=b\circ a=b$.
\end{thm}

For a COSEA $\escript$, we call $a\in S(\escript )$ \textit{one-dimensional} if $a\ne 0$ and if $b\in\escript$ with $b\le a$, then $b=\lambda a$ for some $\lambda\in [0,1]$. We denote the set of one--dimensional elements of $\escript$ by $S_1(\escript )$. It is shown in \cite{gud18} that if $a\in S_1(\escript )$ then there exists an $\ahat\in\Omega (\escript )$ such that $\ahat (a)=1$. A COSEA is \textit{state-unique} if $\ahat$ is unique. Although it is not known whether an arbitrary COEA is state-unique, it is shown in \cite{wet181,wet182} that every COSEA is state-unique.

A \textit{finite context} in a COEA $\escript$ is a finite set $\brac{a_1,\ldots ,a_n}\subseteq S_1(\escript )$ such that
\begin{equation*}
a_1\oplus a_2\oplus\cdots\oplus a_n=1
\end{equation*}
It follows that $\ahat _i(a_j)=\delta _{ij}$. We denote the set of finite contexts in $\escript$ by $\cscript (\escript )$. We interpret a finite context as a finest sharp measurement. We say that $\escript$ is \textit{finite-dimensional} if there does not exist an infinite sequence
$a_i\in S_1(\escript )$ such that $a_1\oplus\cdots\oplus a_n$ is defined for all $n$. Thus, there are no infinite contexts. For simplicity, we assume that the COEA's (and (COSEA's) we consider in this paper are finite-dimensional. If $\escript$ is state-unique and
$a,b\in S_1(\escript )$, we call $\ahat (b)$ the \textit{transition probability} from $a$ to $b$. We say that $\escript$ is \textit{symmetric} if
$\ahat (b)=\bhat (a)$ for all $a,b\in S_1(\escript )$. It is shown in \cite{wet181,wet182} that every COSEA is symmetric.

\begin{lem}    
\label{lem22}
If $\escript$ is state-unique and symmetric, then all contexts in $\escript$ have the same cardinality.
\end{lem}
\begin{proof}
Let $\ascript ,\bscript\in\cscript (\escript )$ with $\ascript =\brac{a_1,\ldots ,a_n}$, $\bscript =\brac{b_1,\ldots ,b_m}$. Then
\begin{equation*}
n=\sum _{i,j}\ahat _i(b_j)=\sum _{i,j}\bhat _j(a_i)=m\qedhere
\end{equation*}
\end{proof}

We say that a COEA $\escript$ is \textit{spectral} if $\escript$ is state-unique and for every $b\in\escript$ there exists a context
$\ascript =\brac{a_1,\ldots ,a_n}$ such that
\begin{equation*}
b=\lambda _1a_1\oplus\cdots\oplus\lambda _na_n
\end{equation*}
$\lambda _i\in [0,1]$, $i=i,\ldots ,n$. We denote the set of such $b\in\escript$ corresponding to a fixed context $\ascript$ by $\escript (\ascript )$. It can be shown that every COSEA is spectral \cite{wet182}. A subset $\fscript$ of a COEA $\escript$ is a \textit{sub-COEA} if $0,1\in\fscript$, $a\in\fscript$ implies $a'\lambda a\in\fscript$ for all $\lambda\in [0,1]$ and if $a,b\in\fscript$ with $a\perp b$, then
$a\oplus b\in\fscript$. A subset $\fscript$ of a COSEA $\escript$ is a \textit{sub-COSEA} if $\fscript$ is a sub-COEA and if $a,b\in\fscript$ implies $a\circ b\in\fscript$. It is clear that if $\escript$ is a COEA (COSEA) then $\escript (\ascript )$ is a sub-COEA (sub-COSEA) for every $a\in\cscript (\escript )$.

We close this section with some examples of COEA's and COSEA's. The first example comes from the quantum formalism. Let $H$ be a complex Hilbert space and let $\escript (H)$ be the set of operators on $H$ satisfying $0\le A\le I$ where we are using the usual operator order. For $A,B\in\escript (H)$ we write $A\perp B$ if $A+B\le I$ and in this case we define $A\oplus B=A+B$. For $\lambda\in [0,1]$ and
$A\in\escript (H)$, $\lambda A\in\escript (H)$ is the usual scalar multiplication for operators. It is easy to check that
$\paren{\escript (H),0,I,\oplus}$ is a COEA which we call a \textit{Hilbertian} COEA. The sharp elements of $\escript (H)$ are the projections on $H$. For $\phi\in H$ with $\phi\ne 0$, we denote the projection onto the one-dimensional subspace generated by $\phi$ as $P(\phi )$. Of course, $P(\phi )=P(\psi )$ if and only $\phi =\alpha\psi$ for some $\alpha\in\complex$, $\alpha\ne 0$. The elements of
$S_1\paren{\escript (H)}$ are precisely the $P(\phi )$, $\phi\in H$, $\phi\ne 0$ and $\escript (H)$ is finite-dimensional if and only if $H$
finite-dimensional. In this case, the contexts of $\escript (H)$ correspond to the orthonormal bases of $H$ so
$\cscript\paren{\escript (H)}$ is infinite if $\dim H\ge 2$. If $A\in S_1\paren{\escript (H)}$ with $A=P(\phi )$ where $\doubleab{\phi}=1$,
then $\capahat$ is the unique state given by $\capahat (B)=\elbows{\phi ,B\phi}$ for all $B\in\escript (H)$. Hence,
$\escript (H)$ is state-unique. It follows from the spectral theorem that $\escript (H)$ is state-unique. Moreover, if $B=P(\psi )$,
$\doubleab{\psi}=1$, then the transition probability becomes
\begin{equation*}
\capahat (B)=\capbhat (A)=\ab{\elbows{\phi ,\psi}}^2
\end{equation*}
so $\escript (H)$ is symmetric. If $\fscript$ is a sub-COEA of $\escript (H)$ for some $H$, we call $\fscript$ a \textit{sub-Hilbertian} COEA. An example of a sub-Hilbertian COEA is a von~Neumann algebra of operators on $H$. These are also spectral and symmetric. For
$A,B\in\escript (H)$ define the product $A\circ B=A^{1/2}BA^{1/2}$ where $A^{1/2}$ is the unique positive square root of $A$. It is shown in \cite{gg02,gl08} that with the product $A\circ B$, $\escript (H)$ becomes a COSEA. We also have that $A\circ B=B\circ A$ if and only if $AB=BA$ \cite{gn01}; that is, $A$ and $B$ commute. We then call $\escript (H)$ a \textit{Hilbertian} COSEA, and any sub-COSEA of
$\escript (H)$ is a \textit{sub-Hilbertian} COSEA. As before, a von~Neumann algebra on $H$ is an example of a sub-Hilbertian COSEA.

Our next example comes from fuzzy probability theory \cite{bug94,gud981}. Let $\Omega ,(\ascript )$ be a measurable space in which singleton sets are measurable and let $\escript (\Omega ,\ascript )$ be the set of measurable functions on $\Omega$ with values in
$[0,1]\subseteq\real$. If we define the operations $\oplus,\lambda f$ and $f\circ g=fg$ analogously as in the previous example,
$\escript (\Omega ,\ascript )$ becomes a COSEA. The elements of $\escript (\Omega ,\ascript )$ are called fuzzy events and we call
$\escript (\Omega ,\ascript )$ a \textit{classical} COSEA. The elements of $S\paren{\escript (\Omega ,\ascript )}$ are the characteristic functions (or equivalently, the sets in $\ascript$) and $S_1\paren{\escript (\Omega ,\ascript )}$ consists of the characteristic functions of the singleton sets (or equivalently, the elements of $\Omega$). Notice that $\escript (\Omega ,\ascript )$ is finite-dimensional if and only if
$\Omega$ is finite and in this case there is only one context. Also, $\escript (\Omega ,\ascript )$ is symmetric and spectral. Conversely, it is shown in \cite{gud18} that if a finite-dimensional COEA (COSEA) $\escript$ has only one context, then $\escript$ is isomorphic to classical COEA (COSEA). We have seen that a classical COEA contains only one context while a quantum (Hilbertian) COEA possesses an infinite number of different contexts. Is there anything between? That is, can a finite-dimensional spectral COEA $\escript$ have a finite number, greater than one, of disjoint contexts \cite{gud18}? The answer to this question is negative. In fact, if $\escript$ has more than one context, then it has uncountably many \cite{jp18}.

\section{Commutants} 
In this section, $\escript$ will denote a finite-dimensional COSEA. For $\fscript\subseteq\escript$, the \textit{commutant} of $\fscript$ is defined as
\begin{equation*}
\fscript '=\brac{b\in\escript\colon b\mid a\quad\hbox{for\ all\ }a\in\fscript}
\end{equation*}
Notice that $\fscript '$ is a sub-COSEA of $\escript$. If $\fscript\subseteq\gscript\subseteq\escript$ then $\gscript '\subseteq\fscript '$. We also have that $\fscript\subseteq\fscript ''$, $\fscript '=\fscript '''$, $\fscript '\cap\gscript '\subseteq (\fscript\cap\gscript )'$ and
$(\fscript\cup\gscript )'\subseteq\fscript '\cup\gscript '$ for all $\fscript ,\gscript\subseteq\escript$. We say that $\fscript\subseteq\escript$ is \textit{commutative} if $a\circ b=b\circ a$ for all $a,b\in\fscript$. Clearly, $\fscript$ is commutative if and only if $\fscript\subseteq\fscript '$. It is shown in \cite{gud18} that $\escript$ is commutative if and only if $\escript$ has only one context and hence is isomorphic to a classical COSEA. We call $\escript '$ the \textit{center} of $\escript$. Thus, $\escript =\escript '$ if and only if $\escript$ is isomorphic to a classical COSEA and $\escript '$ is a commutative sub-COSEA of $\escript$. It is clear that
$\brac{\lambda 1\colon\lambda\in [0,1]}\subseteq\escript '$. We say that $\escript$ is a \textit{factor} if
$\escript '=\brac{\lambda 1\colon\lambda\in [0,1]}$.

We now define the \textit{direct sum} $\escript =\escript _1\oplus\escript _2$ of two COSEA's $(\escript _1,0_1,1_1,\oplus )$,
$(\escript _2,0_2,1_2,\oplus )$. We define $(\escript ,0,1,\oplus )$ by
\begin{equation*}
\escript =\escript _1\times\escript _2=\brac{(a_1,a_2)\colon a_1\in\escript _1,a_2\in\escript _2}
\end{equation*}
with $0=(0_1,0_2)$, $1=(1_1,1_2)$. If $a=(a_1,a_2)$, then $a'=(a'_1,a'_2)$. If $a=(a_1,a_2)$, $b=(b_1,b_2)$ then $a\perp b$ if
$a_1\perp b_1$, $a_2\perp b_2$ and
\begin{equation*} 
a\oplus b=(a_1\oplus b_1,a_2\oplus b_2)
\end{equation*}
For $\lambda =[0,1]$ define $\lambda (a_1,a_2)=(\lambda a_1,\lambda a_2)$ and we define
\begin{equation*} 
(a_1,a_2)\circ (b_1,b_2)=(a_1\circ b_1,a_2\circ b_2)
\end{equation*}
It is easy to check that $\escript$ is a COSEA. We have that $(a_1,a_2)\le (b_1,b_2)$ if and only if $a_1\le b_1$, $a_2\le b_2$ and
\begin{equation*} 
\escript '=\brac{(a_1,a_2)\colon a_1\in\escript '_1,a_2\in\escript '_2}
\end{equation*}
Clearly, $(a_1,a_2)\in S(\escript )$ if and only if $a_1\in S(\escript _1)$ and $a_2\in S(\escript _2)$.

\begin{lem}    
\label{lem31}
Let $\escript =\escript _1\oplus\escript _2$.
{\rm (i)}\enspace $(a_1,a_2)\in S_1(\escript )$ if and only if $a_1=0_1$ and $a_2\in S_1(\escript _2)$ or $a_2=0_2$ and
$a_1\in S_1(\escript _1)$.
{\rm (ii)}\enspace $\ascript\in\cscript (\escript )$ if and only if
\begin{equation*} 
\ascript =\brac{(a_i,0_2),(0_1,b_j)}
\end{equation*}
where $\brac{a_1}\in\cscript (\escript _1)$ and $\brac{b_j}\in\cscript (\escript _2)$
\end{lem}
\begin{proof}
(i)\enspace Necessity is clear. For sufficiency, suppose that $(a_1,a_2)\in S_1(\escript )$ and $a_1\ne 0_1$, $a_2\ne 0_2$. Then
$(a_1,0_2)\le(a_1,a_2)$ but for $\lambda\in [0,1]$ we have that
\begin{equation*} 
(a_1,0_2)\ne (\lambda a_1,\lambda a_2)=\lambda (a_1,a_2)
\end{equation*}
which is a contradiction. Hence, $a_1=0_1$ or $a_2=0_2$. Clearly, if $a_1\ne 0$, then $a_1\in S_1(\escript _1)$ and if $a_2\ne 0$, then
$a_2\in S_1(\escript _2)$.
(ii)\enspace This follows from (i).
\end{proof}

We shall need the following lemma to prove Theorem~3.3.

\begin{lem}    
\label{lem32}
{\rm (i)}\enspace If $a\mid c$ and $a\mid (c\oplus d)$ then $a\mid d$.
{\rm (ii)}\enspace If $c\le b$ and $a\mid c$, $a\mid b$ then $a\mid (b\ominus c)$.
{\rm (iii)}\enspace If $c\le b$ then $b\ominus c=(c\oplus b')'$.
{\rm (iv)}\enspace If $\fscript$ is a sub-COSEA of $\escript$ and $b,c\in\fscript$ with $c\le b$, then $b\ominus c\in\fscript$.
\end{lem}
\begin{proof}
(i)\enspace Let $b=c\oplus d$ so that $a\mid c$ and $a\mid b$. Now $c\oplus d\oplus b'=1$ so $d=(c\oplus b')'$. Since $a\mid b$, $a\mid b'$ and since $a\mid c$ we have that $a\mid c\oplus b'$. Hence, $a\mid d$.
(ii)\enspace Since $c\le b$ we have that $b=c\oplus (b\ominus c)$. Since $a\mid c$ and $a\mid b$, by (i) $a\mid (b\ominus c)$.
(iii)\enspace This follows from (i).
(iv)\enspace Since $b,c\in\fscript$ we have that $b'$ and $c\oplus b'\in\fscript$. Hence, by (iii).
\begin{equation*}
b\ominus c=(c\oplus b')'\in\fscript\qedhere
\end{equation*}
\end{proof}

\begin{thm}    
\label{thm33}
A COSEA $\escript$ is isomorphic to a direct sum of two COSEA's if and only if there exists an $a\in S(\escript )\cap\escript '$ with
$a\ne 0,1$.
\end{thm}
\begin{proof}
If $\escript$ is isomorphic to a direct sum of two COSEA's, we can just as well assume that $\escript =\escript _1\oplus\escript _2$. We then have that $(1_1,0_2)\in S(\escript )\cap\escript '$ and $(1_1,0_2)\ne (1_1,1_2)=1$ and $(1_1,0_2)\ne (0_1,0_2)=0$. Conversely, suppose
$a\in S(\escript )\cap\escript '$ with $a\ne 0,1$. Let
\begin{equation*}
\escript _1=\brac{a\circ b\colon b\in\escript}
\end{equation*}
and define $0_1=a\circ 0=0$ and $1_1=a\circ 1=a$. For $a\circ b\in\escript _1$ define
\begin{equation*}
(a\circ b)'=a\circ b'\in\escript _1
\end{equation*}
Define $a\circ b_1\perp _1a\circ b_2$ if $b_1\perp b_2$ and in this case
\begin{equation*}
a\circ b_1\oplus _1a\circ b_2=a\circ (b_1\oplus b_2)=a\circ b_1\oplus a\circ b_2\in\escript _1
\end{equation*}
It is easy to check that $(\escript _1,0_1,1_1,\oplus _1)$ is an effect algebra. Letting $\lambda (a\circ b)=a\circ (\lambda b)$ makes
$\escript _1$ into a COSEA. Defining
\begin{equation*}
(a\circ b)\circ _1(a\circ b)=(a\circ b)\circ (a\circ c)=a\circ (b\circ c)\in\escript _1
\end{equation*}
we see that $a\circ b\mid _1a\circ c$. We now show that $(\escript _1,0_1,1_1,\oplus _1\circ _1)$ is a COSEA. It is easy to verify that (S1) and (S2) hold. To verify (S3) suppose that $(a\circ b)\circ _1(a\circ c)=0$. Then
\begin{equation*}
a\circ (b\circ c)=(a\circ b)\circ (a\circ c)=0
\end{equation*}
Hence, $a\circ b\mid a\circ c$ so $a\circ b\mid _1a\circ c$. To verify (S4) suppose that $a\circ b\mid _1a\circ c$. Then $a\circ b\mid a\circ c$. Since $a=a\circ c\oplus a\circ c'$ and $a\circ b\mid a$, $a\circ b\mid a\circ c$ it follows from Lemma~\ref{lem32}(i) that $a\circ b\mid a\circ c'$ so $a\circ b\mid (a\circ c)'$. Moreover, for all $d\in\escript$ we have
\begin{equation*}
(a\circ b)\circ\sqbrac{(a\circ c)\circ (a\circ d)}=\sqbrac{(a\circ b)\circ (a\circ c)}\circ (a\circ d)
\end{equation*}
The verification of (S5) and (S6) are straightforward. We conclude that $\escript _1$ is a COSEA. Now $a'\in S(\escript )$ with $a'\ne 0,1$ so letting $\escript _1=\brac{a'\circ b\colon b\in\escript}$ with similar definitions we have that $(\escript _2,0_2,1_2,\oplus _2,\circ _2)$ is a COSEA. Every element of $\escript$ has the unique representation $b=a\circ b\oplus a'\circ b$, $a\circ b\in\escript _1$,
$a'\circ b\in\escript _2$. Defining the map $J\colon\escript\to\escript _1\oplus\escript _2$ by $J(b)=(a\circ b,a'\circ b)$ it is straightforward to show that $J$ is an isomorphism.
\end{proof}

Since $\escript$ is spectral, every $b\in\escript$ has a representation $b=\lambda _1a_1\oplus\cdots\oplus\lambda _aa_n$ for some
$\brac{a_i}\in\cscript (\escript )$, $\lambda _i\in [0,1]$. We denote the set of effects that have such a representation relative to a context
$\ascript\in\cscript (\escript )$ by $\escript (\ascript )$. It is clear that $\escript (\ascript )$ is a commutative sub-COSEA of $\escript$. In fact, if 
$b$ is as above and $c=\mu a_1\oplus\cdots\oplus\mu _na_n$, $\mu _1\in [0,1]$, then $b\perp c$ if and only if $\lambda _i+\mu _i\le 1$, $i=1,\ldots ,n$ and in this case
\begin{equation*}
b\oplus c=(\lambda _1+\mu _1)a_1\oplus\cdots\oplus (\lambda _n+\mu _n)a_n
\end{equation*}
In general, we have
\begin{equation*}
b\circ c=(\lambda _1\mu _1)a_1\oplus\cdots\oplus (\lambda _n\mu _n)a_n
\end{equation*}
In the representation for $b\in\escript$, the $\lambda _i$ need not be distinct but since the sum of sharp elements is sharp, we can write
\begin{equation}                
\label{eq31}
b=\lambda '_1c_1\oplus\cdots\oplus\lambda '_mc_m
\end{equation}
where $c_i\in S(\escript )$, $\lambda '_i\ne\lambda'_j$, $i\ne $. The next result follows from Theorem~4.3 in \cite{gud18}.

\begin{thm}    
\label{thm34}
Any $b\in\escript$ has a unique representation \eqref{eq31} where $\lambda '_i\in [0,1]$, $\lambda '_i\ne\lambda '_j$, $i\ne j$,
$c_i\in S(\escript )$, $c_1\oplus\cdots\oplus c_m=1$ and $c_i\in\brac{b}''$.
\end{thm}

\begin{thm}    
\label{thm35}
In a COSEA $\escript$, $a\mid b$ if and only if $a,b\in\escript (\ascript )$ for some $\ascript\in\cscript (\escript )$.
\end{thm}
\begin{proof}
If $a,b\in\escript (\ascript )$, then clearly $a\mid b$. Conversely, suppose that $a\mid b$. By Theorem~\ref{thm34}, we have
$a=\oplus\lambda _ia_i$, $b=\oplus\mu _ib_i$, $\lambda _i\ne\lambda _j$, $\mu _i\ne\mu _j$, $i\ne j$, $a_i,b_i\in S(\escript )$ and
$\oplus a_i=\oplus b_i=1$. Moreover, by Theorem~\ref{thm34}, $a_i\mid b_j$ for all $i,j$. Then $a_i\circ b_j\in S(\escript )$ and
$\oplus a_i\circ b_j=1$. Letting $e_k$ be the nonzero $a_i\circ b_j$ we have that $e_k\in S(\escript )$ and $\oplus e_k=1$. Then
$a_i=\oplus\brac{e_k\colon e_k\le a_i}$ and similarly for the $b_i$. Reordering the $\lambda _i$ and $\mu _i$ if necessary we can write
$a=\oplus\lambda _ie_i$, $b=\oplus\mu _ie_i$. Finally, we can construct a context $\ascript =\brac{c_k}$ such that $e_i=\oplus c_{k_i}$ for all
$i$. Then $a=\oplus\lambda _ic_i$, $b=\oplus\mu _ic_i$ so that $a,b\in\escript (\ascript )$.
\end{proof}

\begin{lem}    
\label{lem36}
If $a\in S_1(\escript )$, $b\in S(\escript )$, then $a\mid b$ if and only if $a\circ b=0$ or $a\le b$.
\end{lem}
\begin{proof}
If $a\circ b=0$ or $a\le b$, then by Theorem~\ref{thm21}, $a\mid b$. If $a\mid b$, then since $a\circ b\le a$ we have that $a\circ b=\lambda a$ for some $\lambda\in [0,1]$. Since $a\circ b\in S(\escript )$, $\lambda ^2a=\lambda a$ so $\lambda ^2=\lambda$. Hence, $\lambda =0$ or
$\lambda =1$. If $\lambda =0$, then $a\circ b=0$. If $\lambda =1$, then
\begin{equation*}
a=a\circ b=b\circ a\le b\qedhere
\end{equation*}
\end{proof}

\begin{thm}    
\label{thm37}
If $a\in S_1(\escript )$, then
\begin{equation}                
\label{eq32}
\brac{a}'=\brac{b\colon b=\lambda a\oplus\bigoplus\lambda _ia_i,\brac{a,a_1,\ldots ,a_n}\in\cscript (\escript ),\lambda,\lambda _i\in [0,1]}
\end{equation}
\end{thm}
\begin{proof}
If $b=\lambda a\oplus\bigoplus\lambda _ia_i$ as in \eqref{eq32}, then clearly $b\mid a$. Conversely suppose $b\mid a$. By
Theorem~\ref{thm34} we can write $b=\oplus\mu _ic_i$, $c_i\in S(\escript )$, $\mu _i\ne\mu _j$, $\mu _i\ne 0$, $c_i\circ c_j=0$, $i\ne j$. Also by Theorem~\ref{thm34} we have that $a\mid c_i$ for all $i$ so by Lemma~\ref{lem36} $a\circ c_i=0$ or $a\le c$. If $a\circ c_i=0$ for all $i$, then form a context $\brac{a,a_1,\ldots ,a_n}$ such that $b=0a\oplus\bigoplus\lambda _ia_i$. Otherwise, there is a $j$ such that
$a\le c_j$ and $a\circ c_i=0$ for all $i\ne j$. We again form a context $\brac{a,a_1,\ldots ,a_n}$ such that
$b=\lambda a\oplus\bigoplus\lambda _ia_i$.
\end{proof}

\begin{thm}    
\label{thm38}
A COSEA $\escript$ is a factor if and only if $\escript$ is not isomorphic to the direct sum of two COSEA's.
\end{thm}
\begin{proof}
Suppose $\escript$ is a factor. If $\escript$ is isomorphic to a direct sum of COSEA's $\escript _1,\escript _2$, then by Theorem~\ref{thm33} there is an $a\in S(\escript )\cap\escript '$ with $a\ne 0,1$. But then $a=\lambda 1$ for some $\lambda\in (0,1)$. Since $a^2=a$ we have that
$\lambda ^2=\lambda$ so $\lambda =0$ or $\lambda =1$ which is a contradiction. Conversely, suppose $\escript$ is not a factor so that
$\escript '\ne\brac{\lambda 1\colon\lambda\in [0,1]}$. Then there is a $b\in\escript '$ with $b\ne\lambda 1$ for any $\lambda\in [0,1]$. By Theorem~\ref{thm34}, there exists an $a\in S(\escript )\cap\brac{b}''$ with $a\ne 0,1$. Since $\brac{b}'=\escript$ we have that
$a\in\brac{b}''=\escript '$. By Theorem~\ref{thm33}, $\escript$ is isomorphic to the direct sum of two COSEA's.
\end{proof}

For $\fscript\subseteq\escript$, if $a\in\fscript\cap S(\escript )$ with $a\ne 0$, we say that $a$ is \textit{minimal sharp} in $\fscript$ if
$b\in\fscript\cap S(\escript )$ and $b\le a$, then $b=a$.
\begin{thm}    
\label{thm39}
$\fscript$ is a commutative sub-COSEA of $\escript$ if and only if there exist minimal sharp elements $a_1,\ldots ,a_n$ in $\fscript$ such that $a_1\oplus\cdots\oplus a_n=1$ and
\begin{equation}                
\label{eq33}
\fscript =\brac{\lambda _1a_1\oplus\cdots\oplus\lambda _na_n\colon\lambda _i\in [0,1]}
\end{equation}
\end{thm}
\begin{proof}
If $\fscript$ has the form \eqref{eq33}, since $a_i\mid a_j$, $\fscript\subseteq\fscript '$ and it is easy to show that $\fscript$ is a sub-COSEA. Conversely, suppose $\fscript$ is a commutative sub-COSEA of $\escript$. If $b\in\fscript\cap S(\escript )$ with $b\ne 0$ we show there exists a minimal sharp $a$ in $\fscript$ such that $a\le b$. If $b$ is minimal sharp in $\fscript$ we are finished. Otherwise, there exists an
$a_1\in\fscript\cap S(\escript )$ with $a_n\ne 0$ and $a_1<b$. If $a_1$ is minimal sharp in $\fscript$ we are finished. Otherwise, there exists an $a_2\in\fscript\cap S(\escript )$ with $a_2\ne 0$ and $a_2<a_1<b$. This process must end because if $a_1>a_2>a_2>\cdots$ with, $a_i\in\fscript\cap S(\escript )$, $a_i\ne 0$, then letting $b_i=a_i\ominus a_{i+1}$, $i=1,2,\ldots$, we have $b_i\in S(\escript )$ and $b_i\perp b_j$,
$i\ne j$. Since $\escript$ is spectral, there exist $c_i\in S_1(\escript )$ such that $c_i\le b_i$, $i=1,2,\ldots$, but this contradicts the finite-dimensionality of $\escript$. We conclude that for $b\in\fscript\cap S(\escript )$ with $b\ne 0$, there is a minimal sharp $a$ in $\fscript$ such that $a\le b$. Let $a_1,a_2,\ldots ,a_n$ be the minimal sharp elements of $\fscript$. Again, because of finite dimensionality there is a finite number of these. Moreover, we have $a_1\oplus\cdots a_n=1$. If $d\in\fscript$, then Theorem~\ref{thm34} there exist
$d_j\in\fscript\cap S(\escript )$ such that
\begin{equation*}
d=\lambda _1d_1\oplus\cdots\oplus\lambda _md_m
\end{equation*}
where $\lambda _j\in [0,1]$ and $d_1\oplus\cdots\oplus d_m=1$. By our previous work $d_j=\oplus a_{i_j}$ so that
$d=\mu _1a_1\oplus\cdots\oplus\mu _na_n$, $\mu _i\in [0,1]$.
\end{proof}

\begin{cor}    
\label{cor310}
There exist minimal sharp elements $a_1,\ldots ,a_n$ in $\escript '$ such that $a_1\oplus\cdots\oplus a_n=1$ and
\begin{equation*}
\escript '=\brac{\lambda _1a_1\oplus\cdots\oplus\lambda _na_n\colon\lambda _i\in [0,1]}
\end{equation*}
\end{cor}

\begin{lem}    
\label{lem311}
If $a$ is a minimal sharp element of $\escript '$ and $\fscript =\brac{a\circ b\colon b\in\escript}$, then $\fscript$ is a COSEA with unit $a$ and $\fscript$ is a factor.
\end{lem}
\begin{proof}
We have shown in the proof of Theorem~\ref{thm33} that $\fscript$ is a COSEA with unit $a$. To show that $\fscript$ is a factor, we must show that $\fscript '\cap\fscript=\brac{\lambda a\colon\lambda\in [0,1]}$. If $a\circ b\in\fscript '\cap\fscript\cap S(\escript )$, then
$a\circ b\mid a\circ c$ for all $c\in\escript$. We also have that $(a\circ b)\circ (a'\circ c)=0$ so $a\circ b\mid a'\circ c$ for all $c\in\escript$. Since $c=a\circ c\oplus a'\circ c$ we have $a\circ b\mid c$ so $a\circ b\in\escript '$. Since $a\circ b\le a$ and $a$ is minimal sharp in
$\escript$ we conclude that if $b\ne 0$ then $a\circ b=a$. Hence, the only sharp elements of $\fscript '\cap\fscript$ are $0$ and $a$. Since every $c\in\fscript '\cap\fscript$ has the form $c=\lambda _1c_1\oplus\cdots\oplus\lambda _nc_n$, $\lambda _i\in [0,1]$, $c_i\in S(\fscript )$ we have that $c=\lambda a$, $\lambda\in [0,1]$. Therefore, $\fscript$ is a factor.
\end{proof}

We can extend the definition of direct sum to more than two summands. We define
\begin{equation*}
\escript _1\oplus\escript _2\oplus\escript _3=(\escript _1\oplus\escript _2)\oplus\escript _3
\end{equation*}
and of course, the placement of the parenthesis is immaterial. In a similar way, we define
$\escript =\escript _1\oplus\escript\oplus\cdots\oplus\escript _n$. For convenience, write $(a_1,\ldots ,a_n)\in\escript$ as
$a_1\oplus\cdots\oplus a_n$, $a_i\in\escript _i$, $i=1,\ldots ,n$. We then have $a_i\circ a_j=0$, $i\ne j$, and $1_1\oplus\cdots\oplus 1_n=1$. Also,
\begin{equation*}
\escript '=\brac{a_1\oplus\cdots\oplus a_n\colon a_i\in\escript '_i}
\end{equation*}

\begin{thm}    
\label{thm312}
Any finite-dimensional COSEA $\escript$ is isomorphic to the direct sum of a finite number of factors.
\end{thm}
\begin{proof}
By Corollary~\ref{cor310} there exist minimal sharp elements $a_1,\ldots ,a_n$ in $\escript '$ with $a_1\oplus\cdots\oplus a_n=1$. By
Lemma~\ref{lem311}, $\escript _i=\brac{a_i\circ b\colon b\in\escript}$ is a factor with unit $a_i$. Since every $b\in\escript$ has the form
\begin{equation*}
b=a_1\circ b\oplus\cdots\oplus a_n\circ b
\end{equation*}
it follows that $\escript$ is isomorphic to $\escript _1\oplus\cdots\oplus\escript _n$.
\end{proof}

We close this section with a result about the state space of the direct sum. If $V$ is a real vector space and $A_1,\ldots ,A_n\subseteq V$ we define the \textit{convex hull} of $a_1,\ldots ,A_n$ by
\begin{align*}
CH&(A_1,\ldots ,A_n)\\
&=\brac{\lambda _1v_1+\cdots +\lambda _nv_n\colon\lambda _i\le 0,\ \sum\lambda _i=1,\ v_i\in A_i,\ i=1,\ldots ,n}
\end{align*}

\begin{thm}    
\label{thm313}
$\Omega (\escript _1\oplus\cdots\oplus\escript _n)=CH\paren{\Omega (\escript _1),\ldots ,\Omega (\escript _n)}$
\end{thm}
\begin{proof}
We shall show that $\Omega (\escript _1\oplus\escript _2)=CH\paren{\Omega (\escript _1),\Omega (\escript _2)}$ and the general result easily follows. If $\omega _1\in\Omega (\escript _1)$, $\omega _2\in\Omega (\escript _2)$, $\lambda\in [0,1]$,
$(a,b)\in\escript =\escript _1\oplus\escript _2$, define
\begin{equation*}
\omega (a,b)=\lambda\omega _1(a)+(1-\lambda )\omega _2(b)
\end{equation*}
To show that $\omega\in\Omega (\escript )$ we have that
\begin{align*}
\omega (1_1,1_2)&=\lambda\omega _1(1_1)+(1-\lambda )\omega _2(1_2)=1\\
\intertext{and}
\omega\sqbrac{(a_1,a_2)\oplus (b_1,b_2)}&=\omega\sqbrac{(a_1\oplus b_1,a_2\oplus b_2)}\\
  &=\lambda\omega _1(a_1\oplus b_1)+(1-\lambda )\omega _2(a_2\oplus b_2)\\
  &=\lambda\sqbrac{\omega _1(a_1)+\omega _1(b_1)}+(1-\lambda )\sqbrac{\omega _2(a_2)+\omega _2(b_2)}\\
  &=\sqbrac{\lambda\omega _1(a_1)+(1-\lambda )\omega _2(a_2)}+\sqbrac{\lambda\omega _1(b_1)+(1-\lambda )\omega _2(b_2)}\\
  &=\omega (a_1,a_2)+\omega (b_1,b_2)
\end{align*}
Hence, $CH\paren{\Omega (\escript _1),\Omega (\escript _2)}\subseteq\Omega (\escript _1\oplus\escript _2)$. To show that
$\Omega (\escript _1\oplus\escript _2)\subseteq CH\paren{\Omega (\escript _1),\Omega (\escript _2)}$, let
$\omega\in\Omega (\escript _1\oplus\escript _2)$. If $\omega (1_1,0)=0$ then for $b\in\escript _2$ define $\omega _2(b)=\omega (0_1,b)$. Since
$\omega (0_1,1_2)=1$, $\omega _2\in\Omega (\escript _2)$ and we have that
\begin{equation*}
\omega (a,b)=\omega\paren{(a,0_2)\oplus (0_1,b)}=\omega (0_1,b)=\omega _2(b)
\end{equation*}
Similarly, if $\omega (0_1,1_2)=0$, then letting $\omega _1(a)=\omega (a,0_2)$ we have that $\omega (a,b)=\omega _1(a)$. If
$\omega (1_1,0_2)$, $\omega (0_1,1_2)\ne 0$, define $\omega _1\in\Omega (\escript _1)$, $\omega _2\in\Omega (\escript _2)$ by
\begin{equation*}
\omega _1(a)=\frac{1}{\omega (1_1,0_2)}\,\omega (a,0_2),\quad\omega _2(b)=\frac{1}{\omega (0_1,1_2)}\,\omega (0_1,b)
\end{equation*}
Then $\omega (1_1,0_2)+\omega (0_1,1_2)=\omega (1)=1$ and
\begin{equation*}
\omega (a,b)=\omega (a,0_2)+\omega (0_1,b)=\omega (1_1,0_2)\omega _1(a)+\omega (0_1,1_2)\omega _2(b)\qedhere
\end{equation*}
\end{proof}

\section{Conditioning and Spectra} 
As before $\escript$ will denote a finite-dimensional COSEA and if $a\in S_1(\escript )$ then $\ahat$ is the unique state on $\escript$ such that
$\ahat (a)=1$. If $b\in\escript$ and $\omega\in\Omega (\escript )$ with $\omega (b)\ne 0$ we define the \textit{conditional probability for}
$\omega$ \textit{given} $b$ as $\omega (c\mid b)=\omega (b\circ c)/\omega (b)$ for every $c\in\escript$. Notice that $\omega (\ctimes\mid b)$ is indeed a state on $\escript$.

\begin{thm}    
\label{thm41}
Let $a\in S_1(\escript )$.
{\rm (i)}\enspace $\ahat$ is the unique state on $\escript$ such that $a\circ b=\ahat (b)a$ for all $b\in\escript$.
{\rm (ii)}\enspace $\ahat$ is the unique state on $\escript$ such that $\ahat (b)=\ahat (a\circ b)$ for all $b\in\escript$.
{\rm (iii)}\enspace If $\omega\in\Omega (\escript )$ with $\omega (a)\ne 0$, then $\omega (b\mid a)=\ahat (b)$ for all $b\in\escript$.
\end{thm}
\begin{proof}
(i)\enspace Since $a\circ b\le a$, there exists $\lambda _a(b)\in [0,1]$ such that $a\circ b=\lambda _a(b)a$. Applying $\ahat$ to both sides gives
$\lambda _a(b)=\ahat (a\circ b)$. It is clear that $\lambda _a\in\Omega (\escript )$ and $\lambda _a(a)=1$. Hence, $\lambda _a=\ahat$ so that
$a\circ b=\ahat (b)a$ for all $b\in\escript$. If $\omega\in\Omega (\escript )$ satisfies $a\circ b=\omega (b)a$ for all $b\in\escript$, letting $b=a$ gives
\begin{equation*}
a=a\circ a=\omega (a)a
\end{equation*}
Hence, $\omega (a)=1$ so $\omega =\ahat$. Thus, $\ahat$ is unique.
(ii)\enspace By (i) we have that
\begin{equation*}
\ahat (a\circ b)=\ahat (b)\ahat (a)=\ahat (b)
\end{equation*}
for all $b\in\escript$. If $\omega\in\Omega (\escript )$ satisfies $\omega (b)=\omega (a\circ b)$ for all $b\in\escript$, letting $b=1$ gives
$\omega (a)=\omega (1)=1$ so that $\omega =\ahat$.
(iii)\enspace If $\omega (a)\ne 0$, applying (i) gives
\begin{equation*}
\omega (b\mid a)=\frac{\omega (a\circ b)}{\omega (a)}=\frac{\omega\paren{\ahat (b)a}}{\omega (a)}=\ahat (b)\qedhere
\end{equation*}
\end{proof}

From Theorem~\ref{thm41}(iii) we have that $\ahat (b)=\omega (b\mid a)$ for all $\omega\in\Omega (\escript )$ with $\omega (a)\ne 0$. We conclude that $\ahat$ is the \textit{universal} conditional probability given $a$.

Let $\Omegahat (\escript )=\Omega (\escript )\cup\brac{0}$ where $0(b)=0$ for all $b\in\escript$. For all $a\in\escript$ we define the \textit{conditional probability map} $\gamma _a\colon\Omegahat (\escript )\to\Omegahat (\escript )$ by $\gamma _a(0)=0$ and for $\omega\ne 0$
\begin{equation*}
\gamma _a(\omega )=\begin{cases}\omega (\ctimes\mid a)&\hbox{if }\omega (a)\ne 0\\0&\hbox{if }\omega (a)=0\end{cases}
\end{equation*}
It is clear that $\gamma _0(\omega )=0$ and $\gamma _1(\omega )=\omega$ for all $\omega\in\Omegahat (\escript )$. The next result summarizes properties of $\gamma$.

\begin{lem}    
\label{lem42}
{\rm (i)}\enspace If $a\in S_1(\escript )$, the $\ahat$ is the unique nonzero fixed point of $\gamma _a$; that is, $\gamma _a\omega =\omega$,
$\omega\ne 0$ implies that $\omega =\ahat$.
{\rm (ii)}\enspace If $a\perp b$, $c\mid a$, $c\mid b$ then for all $\omega\in\Omega (\escript )$ we have that
\begin{equation}                
\label{eq41}
\omega (a\oplus b)\gamma _{a\oplus b}(\omega )(c)=\omega (a)\gamma _a(\omega )(c)+\omega (b)\gamma _b(\omega )(c)
\end{equation}
{\rm (iii)}\enspace If $a\mid b$, then for all $\omega\in\Omegahat (\escript )$ we have that
\begin{equation}                
\label{eq42}
\omega (a')\gamma _{a'}(\omega )(b)=\omega (b)-\omega (a)\gamma _a(\omega )(b)
\end{equation}
{\rm (iv)}\enspace For all $\omega\in\Omegahat (\escript )$ and $c\in\escript$ we have that
\begin{align}                
\label{eq43}
\omega (a\circ b)\gamma _{a\circ b}(\omega )(c)=\omega\sqbrac{(a\circ b)\circ c}\\
\intertext{and}
\label{eq44}   
\omega (a\circ b)\sqbrac{\gamma _b\gamma _a(\omega )}(c)=\omega\sqbrac{a\circ (b\circ c)}
\end{align}
\end{lem}
\begin{proof}
Conditions \eqref{eq41}--\eqref{eq44} clearly hold if $\omega =0$. We thus assume that $\omega\in\Omega (\escript )$.
(i)\enspace We have from Theorem~\ref{thm41}(ii) that
\begin{equation*}
\gamma _a(\ahat )(b)=\ahat (b\mid a)=\ahat (a\circ b)=\ahat (b)
\end{equation*}
Hence, $\gamma _a(\ahat )=\ahat$. Now if $\gamma _a\omega =\omega$, then $\omega (a)\ne 0$ and for every $b\in\escript$ we have that
\begin{equation*}
\omega (b)=\gamma _a(\omega )(b)=\frac{\omega (a\circ b)}{\omega (a)}
\end{equation*}
We conclude that $\omega (a)=1$ so that $\omega =\ahat$.
(ii)\enspace If $\omega (a\oplus b)=0$, then $\omega (a)=\omega (b)=0$ so both sides of \eqref{eq41} are $0$. If $\omega (a\oplus b)\ne 0$, then
\eqref{eq41} is equivalent to
\begin{align*}
\omega\sqbrac{(a\oplus b)\circ c}&=\omega\sqbrac{c\circ (a\oplus b)}=\omega (c\circ a\oplus c\circ b)\\
&=\omega (c\circ a)+\omega (c\circ b)=\omega (a\circ c)+\omega (b\circ c)
\end{align*}
(iii)\enspace If $\omega (a')=0$, then the left side of \eqref{eq42} is $0$ and the right side is $\omega (b)=\omega (a\circ b)$. But
$b=b\circ a\oplus b\circ a'$ and since $b\circ a'=a'\circ b\le a'$ we have that $\omega (b\circ a')=0$. Hence, $\omega (b)=\omega (a\circ b)$ so the right side is also $0$. If $\omega (a')\ne 0$, then \eqref{eq42} is equivalent to
\begin{align*}
\omega (a'\circ b)&=\omega (b\circ a')=\omega (b)-\omega (b\circ a)=\omega (b)-\omega (a\circ b)\\
&=\omega (b)=\omega (a)\gamma _a(\omega )(b)
\end{align*}
(iv)\enspace If $\omega (a\circ b)=0$, then both sides of \eqref{eq43} are $0$. If $\omega (a\circ b)\ne 0$, then \eqref{eq43} follows directly. Since $b\circ c\le b$, we have that $a\circ (b\circ c)\le a\circ b$. Thus, if $\omega (a\circ b)=0$ then both sides of \eqref{eq44} are $0$. If
$\omega (a\circ b)\ne 0$, then
\begin{equation*}
\omega (a\circ b)\sqbrac{\gamma _b\gamma _a(\omega )}(c)=\frac{\omega (a\circ b)\gamma _a(\omega )(b\circ c)}{\gamma _a(\omega )(b)}
=\omega\sqbrac{a\circ (b\circ c)}\qedhere
\end{equation*}
\end{proof}

If $\omega (a\oplus b)\ne 0$, then \eqref{eq41} shows that on $\brac{a,b}'$ we have that $\gamma _{a\oplus b}$ is a convex combination
\begin{equation*}
\gamma _{a\oplus b}=\frac{\omega (a)}{\omega (a)+\omega (b)}\,\gamma _a+\frac{\omega (b)}{\omega (a)+\omega (b)}\,\gamma _b
\end{equation*}
If $\omega (a')\ne 0$, then \eqref{eq42} implies that on $\brac{a}'$ we have that
\begin{equation*}
\gamma _{a'}=\frac{\omega -\omega (a)\gamma _a(\omega )}{1-\gamma (a)}
\end{equation*}
If $a\mid b$, then \eqref{eq43} and \eqref{eq44} imply that
\begin{equation*}
\gamma _b\gamma _a=\gamma _a\gamma _b=\gamma _{a\circ b}
\end{equation*}

We know that for $a\in S_1(\escript )$ there exists a unique $\omega\in\Omega (\escript )$ such that $\omega (a)=1$. We now consider whether there are other effects with this property.

\begin{thm}    
\label{thm43}
There exists a unique $\omega\in\Omega (\escript )$ for which $\omega (a)=1$ if and only if there is a context $\brac{a_i}$ such that
\begin{equation}                
\label{eq45}
a=a_1\oplus\lambda _2a_2\oplus\cdots\oplus\lambda _na_n
\end{equation}
where $\lambda _i\in\sqparen{0,1}$.
\end{thm}
\begin{proof}
If $a$ has the form \eqref{eq45}, then $\ahat _1(a)=1$. If $\omega\in\Omega (\escript )$ with $\omega (a)=1$, then
\begin{equation*}
\omega (a_1)+\sum _{i=2}^n\lambda _i\omega (a_i)=1
\end{equation*}
If $\omega (a_j)\ne 0$ for some $j=2,\ldots ,n$ then
\begin{equation*}
1=\omega (a_1)+\sum _{i=2}^n\lambda _i\omega (a_i)<\omega (a_1)+\sum _{i=2}^n\omega (a_i)=1
\end{equation*}
which is a contradiction. Hence, $\omega (a_j)=0$, $j=2,\ldots ,n$. We conclude that $\omega (a_1)=1$ so $\omega =\ahat _1$ and $\ahat _1$ is the unique state such that $\ahat _1(a)=1$. Conversely, suppose there exists a unique $\omega \in\Omega (\escript )$ such that $\omega (a)=1$. Let $a=\oplus\lambda _ia_i$ for some $\brac{a_i}\in\cscript (\escript )$, $\lambda _i\in [0,1]$. Since $\omega (a)=1$ we have that
\begin{equation*}
\sum\lambda _i\omega (a_i)=\omega (a)=1
\end{equation*}
If $\omega (a_j)\ne 0$ and $\lambda _j<1$, then
\begin{equation*}
1=\sum\lambda _i\omega (a_i)<\sum\omega (a_i)=1
\end{equation*}
which is a contradiction. Since $\omega (a_j)\ne 0$ for some $j$ we have $\lambda _j=1$ for some $j$. We can assume that $j=1$ and write $a$ in the form \eqref{eq45}. We have that $\lambda _i<1$, $i=2,\ldots ,n$ because if $\lambda _i=1$ then $\ahat _1(a)=\ahat _i(a)=1$ which contradicts the uniqueness of $\omega$.
\end{proof}

\begin{cor}    
\label{cor44}
If $a\in S(\escript )$ , then there exists a unique $\omega\in\Omega (\escript )$ such that $\omega (a)=1$ if and only if $a\in S_1(\escript )$.
\end{cor}

We say that $b\in\escript$ is \textit{dispersion-free} relative to $\omega\in\Omega (\escript )$ if $\omega (b^2)=\omega (b)^2$. Notice that if
$b\in S(\escript )$, then $\omega (b^2)=\omega (b)^2$ if and only if $\omega (b)=0$ or $\omega (b)=1$. This terminology is due to the definition of dispersion as
\begin{equation*}
\omega\sqbrac{\paren{b-\omega (b)1}^2}=\omega (b^2)-\omega (b)^2\ge 0
\end{equation*}
We say that $b$ is \textit{constant almost everywhere} $\omega\sqbrac{{\rm a.e.}(\omega )}$ if $b=\lambda a\oplus c$, $\lambda\in [0,1]$, where
$a\in S(\escript )$, $a\circ c=0$, $\omega (a)=1$.

\begin{thm}    
\label{thm45}
An effect $b$ is dispersion-free relative to $\omega\in\Omega (\escript )$ if and only if $b$ is constant {\rm a.e.}$(\omega )$.
\end{thm}
\begin{proof}
If $b$ is constant a.e.$(\omega )$, then $b=\lambda a\oplus c$, $a\in S(\escript )$, $a\circ c=0$, $\omega (a)=1$. Then $a\mid c$ and we have that $b^2=\lambda ^2 a\oplus c^2$. Since
\begin{equation*}
a=a\circ c\oplus a\circ c'=a\circ c'=c'\circ a\le c'
\end{equation*}
we have that $1=\omega (a)\le\omega (c')$. Hence, $\omega (c')=1$ so that $\omega (c)=0$. Since $c^2\le c$ and $\omega (c)=0$ we conclude that $\omega (c^2)=0$. Hence,
\begin{equation*}
\omega (b^2)=\lambda ^2\omega (a)=\lambda ^2=\omega (b)^2
\end{equation*}
Conversely, suppose $\omega (b^2)=\omega (b)^2$. Let $b=\lambda _1a_1\oplus\cdots\oplus\lambda _na_n$, $\lambda _i\in [0,1]$,
$\brac{a_i}\in\cscript (\escript )$. Define the random variable $f(a_i)=\lambda _i$ with distribution $\omega (a_i)$. Then the expectation of $f$ becomes
\begin{align*}
E_\omega (f)&=\sum\lambda _i\omega (a_i)=\omega (b)\\
\intertext{and}
E_\omega (f^2)&=\sum\lambda _i^2\omega (a_i)=\omega (b^2)=\omega (b)^2=E_\omega (f)^2
\end{align*}
Hence,
\begin{equation*}
E_\omega\sqbrac{\paren{f-E_\omega (f)}^2}=E_\omega (f^2)-E(f)^2=0
\end{equation*}
Since $\paren{f-E_\omega (f)}^2\ge 0$, $f=E_\omega (f){\rm a.e.}(\omega )$. Therefore,
\begin{equation*}
f(a_i)=E_\omega (f)=\omega (b){\rm a.e.}(\omega )
\end{equation*}
We can assume that
\begin{equation*}
f(a_1)=\cdots =f(a_m)=\omega (b)
\end{equation*}
and $\omega (a_{m+1})=\cdots =\omega (a_n)=0$. Letting $a=a_1\oplus\cdots\oplus a_n$ and
\begin{equation*}
c=\lambda _{m+1}a_{m+1}\oplus\cdots\oplus\lambda _na_n
\end{equation*}
we have that $b=\omega (b)a\oplus c$ where $a\in S(\escript )$, $a\circ c=0$, $\omega (a)=1$.
\end{proof}

It follows from the proof of Theorem~\ref{thm45} that if $a$ is constant a.e.$(\omega )$ then the constant is $\omega (a)$.

We say that $b\in\escript$ has \textit{eigeneffect} $a\in S_1(\escript )$ if $b\mid a$. Notice that $b\mid a$ if and only if $b\circ a=\ahat (b)a$. We call $\ahat (b)$ the \textit{eigenvalue} corresponding to eigeneffect $a$. The set of eigeneffects for $b$ is the \textit{eigenspace} $S_1(b)$ and the set of eigenvalues for $b$ is the \textit{spectrum} $\sigma (b)$. Since $\escript$ is spectral, every $b\in\escript$ can be written as
$b=\lambda _1a_1\oplus\cdots\oplus\lambda _na_n$, $\lambda _i\in [0,1]$, $\brac{a_i}\in\cscript (\escript )$. Since $b\mid a_i$, it follows that
$a_i\in S_1(b)$ and $\lambda _i=\ahat _i(b)\in\sigma (b)$, $i=1,\ldots ,n$. The different eigenvalues of $b$ are unique but there may be various eigeneffects corresponding to the same eigenvalues. For example, if $\lambda _1=\lambda _2$, then $a_1,a_2\in S_1(\escript )$ correspond to
$\lambda _1$. More generally, in this case if $c\in S_1(\escript )$ and $c\le a_1\oplus a_2$ then $c$ corresponds to $\lambda _1$. It is also clear that if $a,b\in S_1(b)$ correspond to different eigenvalues, then $a\circ b=0$. Moreover, $b\in S(\escript )$ if and only if
$\sigma (b)\subseteq\brac{0,1}$ and $b\in S_1(\escript )$ if and only if $1\in\sigma (b)$ and $S_1(b)=\brac{b}$.

We define $m(b)=\min\brac{\lambda\colon\lambda\in\sigma (b)}$ and $M(b)=\max\brac{\lambda\colon\lambda\in\sigma (b)}$. Of course,
$0\le m(b)\le M(b)\le 1$. We define the \textit{numerical range} $r(b)=\sqbrac{m(b),M(b)}$ and the \textit{norm} $\doubleab{b}=M(b)$. It is clear that
$\sigma (\lambda b)=\lambda\sigma (b)$, $r(\lambda b)=\lambda r(b)$ and $\doubleab{\lambda b}=\lambda\doubleab{b}$ for all $b\in\escript$,
$\lambda\in[0,1]$.

\begin{lem}    
\label{lem46}
$r(b)=\brac{\omega (b)\colon\omega\in\Omega (\escript )}$
\end{lem}
\begin{proof}
Let $a_1,a_2\in S_1(b)$ with $b\circ a_1=m(b)a_1$ and $b\circ a_2=M(b)a_2$. For $\lambda\in [0,1]$ we define $\omega _\lambda\in\Omega (\escript )$ by $\omega _\lambda=\lambda\ahat _1+(1-\lambda )\ahat _2$. We then have
\begin{align*}
r(b)&=\sqbrac{m(b),M(b)}=\brac{\lambda m(b)+(1-\lambda )M(b)\colon\lambda\in [0,1]}\\
  &=\brac{\lambda\ahat _1(b)+(1-\lambda )\ahat _2(b)\colon\lambda\in [0,1]}=\brac{\omega _\lambda (b)\colon\lambda\in [0,1]}\\
  &\subseteq\brac{\omega (b)\colon\omega\in\Omega (\escript )}
\end{align*}
Conversely, if $b=\lambda _1a_1\oplus\cdots\oplus\lambda _na_n$, $\brac{a_i}\in\cscript (\escript )$ then $\sigma (b)=\brac{\lambda _i}$.
If $\omega\in\Omega (\escript )$, then $\omega (b)=\sum\lambda _i\omega (a_i)$. Since $\sum\omega (a_i)=1$ we have that
\begin{equation*}
m(b)=\sum m(b)\omega (a_i)\le\sum\lambda _i\omega (a_i)\le\sum M(b)\omega (a_i)=M(b)
\end{equation*}
Hence $m(b)\le\omega (b)\le M(b)$ and we conclude that
\begin{equation*}
\brac{\omega (b)\colon\omega\in\Omega (\escript )}\subseteq r(b)\qedhere
\end{equation*}
\end{proof}

\begin{thm}    
\label{thm47}
{\rm (i)}\enspace $\doubleab{b}=\max\brac{\omega (b)\colon\omega\in\Omega (\escript )}$.
{\rm (ii)}\enspace If $b_1\perp b_2$ then
\begin{equation*}
\doubleab{b_1\oplus b_2}\le\doubleab{b_1}+\doubleab{b_2}
\end{equation*}
{\rm (iii)}\enspace $\doubleab{b}=0$ if and only if $b=0$
{\rm (iv)}\enspace If $a\le b$ then $\doubleab{a}\le\doubleab{b}$ and for all $a\in\escript$, $a\le\doubleab{a}1$.
{\rm (v)}\enspace $\doubleab{a\circ b}\le\doubleab{a}\,\doubleab{b}$.
\end{thm}
\begin{proof}
(i) follows from Lemma~\ref{lem46}.
(ii)\enspace By (i) we have that 
\begin{align*}
\doubleab{b_1\oplus b_2}&=\max\brac{\omega (b_1\oplus b_2)\colon\omega\in\Omega (\escript )}
=\max\brac{\omega (b_1)+\omega (b_2)\colon\omega\in\Omega (\escript )}\\
&\le\max\brac{\omega (b_1)\colon\omega\in\Omega (\escript )}+\max\brac{\omega (b_2)\colon\omega\in\Omega (\escript )}\\
&=\doubleab{b_1}+\doubleab{b_2}
\end{align*}
(iii)\enspace We have that $b=0$ if and only if $\sigma (b)=\brac{0}$ which is equivalent to $\doubleab{b}=0$.
(iv)\enspace If $a\le b$, then there exists a $c\in\escript$ such that $b=a\oplus c$. Hence, for all $\omega\in\Omega (\escript )$ we have that
\begin{equation*}
\omega (a)\le\omega (a)+\omega (c)=\omega (b)
\end{equation*}
It follows from (i) that $\doubleab{a}\le\doubleab{b}$. Since $a=\lambda _1a_1\oplus\cdots\oplus\lambda _na_n$, $\brac{a_i}\in\cscript (\escript )$,
$\sigma (a)=\brac{\lambda _i}$ we have that
\begin{equation*}
a=\lambda _1a_2\oplus\cdots\oplus\lambda _na_n\le M(a)(a_1\oplus\cdots\oplus a_n)=M(a)1=\doubleab{a}1
\end{equation*}
(v)\enspace By (iv) we have $b\le\doubleab{b}1$ and hence, $a\circ b\le\doubleab{b}a$. Again by (iv) we conclude that
\begin{equation*}
\doubleab{a\circ b}\le\doubleab{\,\doubleab{b}a}=\doubleab{a}\doubleab{b}\qedhere
\end{equation*}
\end{proof}

\section{Representation Theorems} 
Let $\escript$ be a finite-dimensional spectral COSEA. For $\ascript =\brac{a_i}\in\cscript (\escript )$ define the complex linear space
\begin{equation*}
\hscript (\ascript )=\brac{\sum\alpha _i\ahat _i\colon\alpha _i\in\complex}
\end{equation*}
For $x,y\in\hscript (\ascript )$ with $x=\sum\alpha _i\ahat _i$, $y=\sum\beta _1\ahat _i$ define the inner product
$\elbows{x,y}=\sum\alphabar _i\beta _i$. Thus, $\hscript (\ascript )$ is a complex Hilbert space that we call the \textit{state space for context}
$\ascript$. Of course, $\hscript (\ascript )$ has orthonormal basis $\ascripthat =\brac{\ahat _i\colon i=1,\ldots ,n}$ and $\dim\hscript (\ascript )=n$.
The elements of $\ascripthat$ can be thought of as states in $\Omega (\escript )$ or as unit vectors in $\hscript (\ascript )$ which again correspond to Hilbert space pure states. We now show that this dual role is consistent. For $b\in\escript$ define the linear operator $L_b$ on
$\hscript (\ascript )$ by $L_b=\sum\ahat _j(b)P(\ahat _j)$. Notice that $L_b$ is a positive operator, $L_0=0$, $L_1=I$, $L_{b'}=I-L_b$ and if
$b\perp c$ then $L_{b\oplus c}=L_b+L_c$. We then have that
\begin{equation*}
\elbows{\ahat _i,L_b\ahat _i}=\elbows{\ahat _i,\sum\ahat _j(b)P(\ahat _j)\ahat _i}=\elbows{\ahat _i,\ahat _i(b)\ahat _i}=\ahat _i(b)
\end{equation*}
so the dual roles are consistent. It is easy to see that $L\colon\escript\to\escript\paren{\hscript (\ascript )}$ need not be injective or surjective and does not preserve sharpness. Moreover, all the $L_b$, $b\in\escript$, commute so they do not convey quantum interference. One can say that $L$ gives a distorted partial view of $\escript$. The reason for this is that we are only employing a single context $\ascript$. Each context gives a partial view and in order to obtain a total view, they must all be considered.

In order to consider several contexts together, we introduce a method to compare them. We say that $\escript$ is \textit{comparable} if for every
$\ascript ,\bscript\in\cscript (\escript )$ there exists a unitary operator $U_{\ascript\bscript}\colon\hscript (\ascript )\to\hscript (\bscript )$ such that
\begin{equation}                
\label{eq51}
\ab{\elbows{U_{\ascript\bscript}\ahat ,\bhat}}^2=\ahat (b)
\end{equation}
for all $a\in\ascript$, $b\in\bscript$ and
\begin{equation}                
\label{eq52}
U_{\bscript\cscript}U_{\ascript\bscript}=U_{\ascript\cscript}
\end{equation}
for all $\cscript\in\cscript (\escript )$. Notice that if $\escript$ is comparable, then any two contexts in $\escript$ have the same cardinality.

We now justify why we assume that $\hscript (\ascript )$ is a complex Hilbert space instead of a real space which may seem to be more natural. In many situations, there is an underlying symmetry group that we would like to represent on $\escript$. This is most accurately accomplished by employing a unitary representation of the group on $\hscript (\ascript )$ for some $\ascript\in\cscript (\escript )$. For a unitary representation, we need $\hscript (\ascript )$ to be complex. Moreover, it is desirable for the representation to be context independent. This motivates requiring that
$\escript$ is comparable because in this case the representations for different contexts are unitarily equivalent.

\begin{lem}    
\label{lem51}
If $\escript$ is comparable, then
{\rm (i)}\enspace $U_{\ascript\ascript}=I$,
{\rm (ii)}\enspace $U_{\ascript\bscript}=U_{\bscript\ascript}^*$,\newline
{\rm (iii)}\enspace $\ab{\elbows{U_{\ascript\bscript}\ahat,U_{\cscript\bscript}\chat\,}}^2=\ahat (c)$.
\end{lem}
\begin{proof}
(i)\enspace Applying \eqref{eq52} gives $U_{\ascript\ascript}U_{\ascript\ascript}=U_{\ascript\ascript}$ Multiplying by $U_{\ascript\ascript}^*$ gives $U_{\ascript\ascript}=I$.
(ii)\enspace By \eqref{eq52} we have that
\begin{equation*}
U_{\bscript\ascript}U_{\ascript\bscript}=U_{\ascript\ascript}=I
\end{equation*}
Hence, $U_{\ascript\bscript}=U_{\bscript\ascript}^*$.
(iii)\enspace Applying \eqref{eq51}, \eqref{eq52} and (ii) gives
\begin{align*}
\ab{\elbows{U_{\ascript\bscript}\ahat,U_{\cscript\bscript}\chat\,}}^2
&=\ab{\elbows{U_{\cscript\bscript}^*U_{\ascript\bscript}\ahat ,\chat\,}}^2
=\ab{\elbows{U_{\bscript\cscript}U_{\ascript\bscript}\ahat ,\chat\,}}^2\\
&=\ab{\elbows{U_{\ascript\cscript}\ahat ,\chat\,}}^2=\ahat (c)\qedhere
\end{align*}
\end{proof}

For $b\in\escript (\bscript )$ with $b=\lambda _1b_1\oplus\cdots\oplus\lambda _nb_n$, define $\btilde\in\hscript (\bscript )$ by
$\btilde =\sum\lambda _iP(\,\bhat _i)$. For comparable $\escript$ define
$\caputilde _{\bscript\ascript}\colon\escript\paren{\hscript (\bscript)}\to\escript\paren{\hscript (\ascript )}$ by
\begin{equation*}
\caputilde _{\bscript\ascript}=U_{\bscript\ascript}BU_{\ascript\bscript}=U_{\bscript\ascript}BU_{\bscript\ascript}^*
\end{equation*}
We say that $\escript$ is \textit{strongly comparable} if $\escript$ is comparable and if $b_1\perp b_2$ with $b_1\in\escript (\ascript )$,
$b_2\in\escript (\bscript )$, $b_1\oplus b_2\in\escript (\cscript )$, then
\begin{equation}                
\label{eq53}
(b_1\oplus b_2)^\sim =\caputilde _{\ascript\cscript}\btilde _1\oplus\caputilde _{\bscript\cscript}\btilde _2
\end{equation}
We see that \eqref{eq53} is a reasonable requirement which postulates that $\oplus$ is independent of its Hilbert space representation.

\begin{thm}    
\label{thm52}
A finite-dimensional COEA $\escript$ is isomorphic to a finite-dimensional Hilbertian sub-COEA if and only if $\escript$ is spectral and strongly comparable.
\end{thm}
\begin{proof}
Suppose $\escript$ is isomorphic to a Hilbertian sub-COEA $\fscript$. For simplicity we can assume that $\escript =\fscript$. It is clear that
$\fscript$ is state-unique. By the spectral theorem, if $b\in\fscript$, then $b=\sum\lambda _ia_i$ where $a_i\in S_1\paren{\escript (H)}$ are polynomial functions of $b$. Hence, $a_i\in\fscript$ so $\escript$ is spectral. To show that $\escript$ is comparable, let $\ascript =\brac{a_i}$,
$\bscript =\brac{b_i}$ be contexts in $\escript$. Then $\brac{\ahat _i}$, $\brac{\bhat _i}$ are orthonormal bases of $\hscript$. Define
$U_{\ascript\bscript}\colon\hscript (\ascript )\to\hscript (\bscript )$ by
\begin{equation*}
U_{\ascript\bscript}\ahat _i=\sum _j\elbows{\bhat _j,\ahat _i}\bhat _j
\end{equation*}
and extend by linearity. It is clear that $U_{\ascript\bscript}$ is unitary. Also, \eqref{eq51} holds because
\begin{equation*}
\ab{\elbows{U_{\ascript\bscript}\ahat _i,\bhat _j}}^2=\ab{\elbows{\bhat _j,\ahat _i}}^2=\ahat _i(b_j)
\end{equation*}
If $\cscript =\brac{c_i}$ is another context, we have that
\begin{align*}
U_{\bscript\cscript}U_{\ascript\bscript}\ahat _i&=\sum _j\elbows{\bhat _j,\ahat _i}U_{\bscript\cscript}\bhat _j
  =\sum _{j,k}\elbows{\bhat _j,\ahat _i}\elbows{\chat _k,\bhat _j}\chat _k\\
    &=\sum _k\elbows{\chat _k,\ahat _i}\chat _k=U_{\ascript\cscript}\ahat _i
\end{align*}
Hence, \eqref{eq52} holds so $\escript$ is comparable. In this case, if $a\in\escript$ then $a=\atilde$ and
$\caputilde _{\ascript\cscript}=\caputilde _{\bscript\cscript}=I$ so clearly $\escript$ is strongly comparable.

Conversely, suppose $\escript$ is spectral and strongly comparable. Fix $\ascript\in\cscript (\escript )$ and let\linebreak 
$J\colon\escript\to\escript\paren{\hscript (\ascript )}$ be defined by
\begin{equation*}
J(b)=\caputilde _{\bscript\ascript}(\btilde )
\end{equation*}
where $b\in\escript (\bscript )$, $\bscript\in\cscript (\escript )$. We first show that $J(b)$ is well-defined. That is, we need to show $J(b)$ is independent of the context $\bscript$ containing $b$. Suppose $b\in\escript (\bscript )\cap\escript (\cscript )$. Letting $b_1=0$, $b_2=b$, we have that $b_1\in\escript (\bscript )$, $b_2\in\escript (\bscript )$ and $b_1\oplus b_2=b\in\escript (\cscript )$. By \eqref{eq53} we have that
\begin{equation*}
\btilde =(b_1\oplus b_2)^\sim=\caputilde _{\bscript\cscript}(\btilde _1)\oplus\caputilde _{\bscript\cscript}(\btilde _2)=\caputilde _{\bscript\cscript}(\btilde )
\end{equation*}
Therefore
\begin{equation*}
\caputilde _{\cscript\ascript}(\btilde )=\caputilde _{\cscript\ascript}\caputilde _{\bscript\cscript}(\btilde )=\caputilde _{\bscript\ascript}(\btilde )
\end{equation*}
Hence, $J(b)$ is well-defined. We now show that $J$ is injective. Let $b\in\escript (\bscript )$ with
$b=\lambda _1b_1\oplus\cdots\oplus\lambda _nb_n$, $c\in\escript (\cscript )$ with $c=\mu _1c_1\oplus\cdots\oplus\mu _nc_n$ and suppose that $J(b)=J(c)$. Then $\caputilde _{\bscript\ascript}(\btilde )=\caputilde _{\cscript\ascript}(\ctilde )$ or equivalently
\begin{equation*}
U_{\bscript\ascript}\btilde U_{\ascript\bscript}=U_{\cscript\ascript}\ctilde U_{\ascript\cscript}
\end{equation*}
This implies that
\begin{equation*}
\btilde =U_{\ascript\bscript}U_{\cscript\ascript}\ctilde U_{\ascript\cscript}U_{\bscript\ascript}=U_{\cscript\bscript}\ctilde U_{\bscript\cscript}
\end{equation*}
which gives $U_{\bscript\cscript}\btilde =\ctilde U_{\bscript\cscript}$. We conclude that
\begin{equation*}
\ctilde (U_{\bscript\cscript}\bhat _i)=(U_{\bscript\cscript}\btilde )\bhat _i=\lambda _iU_{\bscript\cscript}\bhat _i
\end{equation*}
Hence, $U_{\bscript\cscript}\bhat _i$ an eigeneffect of $\ctilde$ with corresponding eigenvalue $\lambda _i$. But the eigenvalues of $\ctilde$ are
$\mu _j$ with corresponding eigeneffects $\chat _j$. Therefore, $\lambda _i=\mu _j$ for some $j$ and $U_{\bscript\cscript}\bhat _i=\chat _j$. Since
\begin{equation*}
\bhat _i(c_k)=\ab{\elbows{U_{\bscript\cscript}\bhat _i,\chat _k}}^2=\ab{\elbows{\chat _j,\chat _k}}^2=\delta _{ij}
\end{equation*}
we have that $\bhat _i(c_j)=1$. We conclude that $b_j=c_i$ for all $i$ so $b=c$. We now show that $J(b_1\oplus b_2)=J(b_1)\oplus J(b_2)$. Suppose that $b_1\perp b_2$ with $b_1\in\escript (\bscript )$, $b_2\in\escript (\dscript )$, $b_1\oplus b_2\in\escript (\cscript )$. By strong comparability we have that
\begin{equation*}
(b_1\oplus b_2)^\sim =\caputilde _{\bscript\cscript}(\btilde _1)\oplus\caputilde _{\dscript\cscript}(\btilde _2)
\end{equation*}
Hence,
\begin{align*}
J(b_1\oplus b_2)&=\caputilde _{\cscript\ascript}\sqbrac{(b_1\oplus b_2)^\sim}
=\caputilde _{\cscript\ascript}\sqbrac{\caputilde _{\bscript\cscript}(\btilde _1)\oplus\caputilde _{\dscript\cscript}(\btilde _2)}\\
&=\caputilde _{\cscript\ascript}\caputilde _{\bscript\cscript}(\btilde _2)\oplus\caputilde _{\cscript\ascript}\caputilde _{\dscript\cscript}(\btilde _2)
=\caputilde _{\bscript\ascript}(\btilde _1)\oplus\caputilde _{\dscript\ascript}(\btilde _2)\\
&=J(b_1)\oplus J(b_2)
\end{align*}
If $\lambda\in [0,1]$, $b\in\escript (\bscript )$, then
\begin{equation*}
J(\lambda b)=\caputilde _{\bscript\ascript}\paren{(\lambda b)^\sim}=\caputilde _{\bscript\ascript}(\lambda\btilde )
=\lambda\caputilde _{\bscript\ascript}(\btilde )=\lambda J(b)
\end{equation*}
It is easy to check that the range of $J$ is a sub-COEA of $\escript\paren{\hscript (\ascript )}$.
\end{proof}

We now consider representations of a finite-dimensional COSEA $\escript$. We first need some preliminary lemmas. We saw in
Theorem~\ref{thm34} that any $a\in\escript$ with $a\ne 0$ has a unique representation $a=\lambda _1c_1\oplus\cdots\oplus\lambda _nc_n$,
$\lambda _i\ne 0$, $\lambda _i\ne\lambda _j$, $i\ne j$, and $c_i\in S(\escript )$. We denote by $\ceils{a}$ the smallest sharp element that dominates $a$.

\begin{lem}    
\label{lem53}
$\ceils{a}$ exists and $\ceils{a}=c_1\oplus\cdots\oplus c_n$.
\end{lem}
\begin{proof}
Let $c=c_1\oplus\cdots\oplus c_n$. Then $c\in S(\escript )$ and $a\le c$. Suppose $b\in S(\escript )$ and $a\le b$. Then $a\circ b=b\circ a=a$. Hence, $b\mid c_i$ and
\begin{equation}                
\label{eq54}
\lambda _1c_1\oplus\cdots\oplus\lambda _nc_n=a=a\circ b=\lambda _1c_1\circ b\oplus\cdots\oplus\lambda _nc_n\circ b
\end{equation}
Now $c_i\circ b\le c_i$ and if $c_i\circ b<c_i$ we would contradict \eqref{eq54}. Hence, $c_i\circ b=c_i$ so that
\begin{equation*}
c\circ b=\oplus (c_i\circ b)=\oplus c_i=c
\end{equation*}
It follows that $c\le b$ so that $c=\ceils{a}$.
\end{proof}

We say that $a\in\escript$ is \textit{pseudo-invertible} if there exists a $b\in\escript$ such that $\ceils{b}=\ceils{a}$, $\doubleab{b}=1$ and
\begin{equation*}
a\circ b=b\circ a=\lambda\ceils{a}
\end{equation*}
for some $\lambda\in [0,1]$. We then call $b$ a \textit{pseudo-inverse} for $a$. (A slightly different definition as well as a version of the next lemma are given in \cite{wet181}.) We denote the smallest nonzero eigenvalue of $a$ by $\lambda (a)$.

\begin{lem}    
\label{lem54}
If $a\ne 0$, then $a$ has a unique pseudo-inverse and $\lambda =\lambda (a)$.
\end{lem}
\begin{proof}
If $a\ne 0$, as before $a$ has the unique representation $a=\lambda _1c_1\oplus\cdots\oplus\lambda _nc_n$, $\lambda _i\ne 0$,
$\lambda _i\ne\lambda _j$, $i\ne j$, $c_i\in S(\escript )$. Letting
\begin{equation*}
b=\lambda (a)\paren{\frac{1}{\lambda _1}\,c_1\oplus\cdots\oplus\frac{1}{\lambda _n}\,c_n}
\end{equation*}
we have from Lemma~\ref{lem53} that
\begin{equation*}
a\circ b=b\circ a=\lambda (a)(c_1\oplus\cdots\oplus c_n)=\lambda (a)\ceils{a}
\end{equation*}
Moreover, $\doubleab{b}=1$, $\ceils{b}=\ceils{a}=c_1\oplus\cdots\oplus c_n$. For uniqueness, suppose $\ceils{d}=\ceils{a}$, $\doubleab{d}=1$ and $a\circ d=d\circ a=\lambda\ceils{a}$. Then $d=\mu _1c_1\oplus\cdots\oplus\mu _nc_n$ and
\begin{equation*}
\mu _1\lambda _1c_1\oplus\cdots\oplus\mu _n\lambda _nc_n=a\circ d=\lambda\ceils{a}
\end{equation*}
This implies that $\mu _1\lambda _i=\lambda$ for all $i$. Hence, $\mu _i=\lambda/\lambda _i$. Since $\doubleab{d}=1$ we have $M(d)=1$ which implies that
\begin{equation*}
\frac{\lambda}{\lambda (a)}=\frac{\lambda}{\min (\lambda _i}=\max\paren{\frac{\lambda}{\lambda _i}}=\max (\mu _i)=\doubleab{d}=1
\end{equation*}
Therefore, $\lambda (a)=\lambda$ and $\mu _i=\lambda (a)/\lambda _i$ so $d=b$.
\end{proof}

We denote the unique pseudo-inverse of $a$ by $a^{-1}$. If $a\ne 0$, $\mu >0$ and $\mu a\in\escript$, then it is easy to show that
$(\mu a)^{-1}=a^{-1}$. It follows that $(a^{-1})^{-1}=a/\doubleab{a}$ and $\paren{(a^{-1})^{-1}}^{-1}=a^{-1}$. We can interpret $a^{-1}$ operationally as the effect that reverses $a$ without interference but with a reduction of intensity by a factor $\lambda (a)$. If $\ceils{a}=1$, we say that $a$ is \textit{invertible} and $a^{-1}$ is the \textit{inverse} of $a$. We say that $\escript$ is \textit{inverse-preserving} if whenever $a$ and $b$ are invertible, then $a\circ b$ is as well and $(a\circ b)^{-1}=a^{-1}\circ b^{-1}$. Notice that the order of $a^{-1}$ and $b^{-1}$ on the right is a bit unexpected but this is the correct order for a sequential product $a\circ b$ in which $a$ is measured first. It is clear that a classical COSEA is inverse-preserving. That a Hilbertian sub-COSEA is also will be shown in Theorem~5.6.

\begin{lem}    
\label{lem55}
{\rm (i)}\enspace $a\in\escript$ is invertible if and only if $a$ does not have a zero eigenvalue.
{\rm (ii)}\enspace If $a\perp b$ and $a$ is invertible then $a\oplus b$ is invertible.
\end{lem}
\begin{proof}
(i)\enspace If $0\in\sigma (a)$, then $\ceils{a}\ne 1$ so $a$ is not invertible. If $0\notin\sigma (a)$, then $\ceils{a}=1$ so $a$ is invertible.
(ii)\enspace If $a$ is invertible, the $\ceils{a}=1$. Suppose $a\oplus b$ is not invertible. Then $\ceils{a\oplus b}\ne 1$ so there exists a $c\in S_1(\escript$ such that
\begin{equation*}
c\circ a\oplus c\circ b=c\circ (a\oplus b)=(a\oplus b)\circ c=0
\end{equation*}
Hence, $c\circ a=0$ which contradicts $\ceils{a}=1$.
\end{proof}

When we consider a sub-Hilbertian COSEA $\fscript\subseteq\escript (H)$ we are assuming the standard sequential product
$A\circ B=A^{1/2}BA^{1/2}$ on $\fscript$.

\begin{thm}    
\label{thm56}
A finite-dimensional COSEA $\escript$ is isomorphic to a finite-dimensional sub-Hilbertian COSEA $\fscript\subseteq\escript (H)$ if and only if
$\escript$ is strongly comparable and inverse-preserving.
\end{thm}
\begin{proof}
Suppose $\escript$ is COSEA isomorphic to $\fscript\subseteq\escript (H)$. For simplicity, we can assume that $\escript =\fscript$. We have shown in Theorem~\ref{thm52} that $\escript$ is strongly comparable. To show that $\escript$ is inverse-preserving, suppose that $A,B\in\escript$ are invertible. It follows from Lemma~\ref{lem55}(i)) that $A$ and $B$ are invertible in the usual operator sense. To avoid confusion, denote the usual operator inverse of $A$ by $\capahat$. We then have that
\begin{equation}                
\label{eq55}
(A\circ B)^\wedge =(A^{1/2}BA^{1/2})^\wedge =\capahat ^{1/2}\capbhat\capahat ^{1/2}=\capahat\circ\capbhat
\end{equation}
Writing $A^{-1}$ as we previously define it we have that
\begin{equation*}
A\circ A^{-1}=A^{-1}\circ A=\lambda (A)I
\end{equation*}
Therefore, $A^{-1}=\lambda (A)\capahat$. Hence, $\lambda (A)\capahat\in\escript$ although $\capahat\notin\escript$ in general. Similarly, $B^{-1}=\lambda (B)\capbhat\in\escript$ and we can rewrite \eqref{eq55} as
\begin{equation*}
A^{-1}\circ B^{-1}=\lambda (A)\lambda (B)\capahat\circ\capbhat =\lambda (A)\lambda (B)(A\circ B)^\wedge =(A\circ B)^{-1}
\end{equation*}
Hence, $(A\circ B)^{-1}$ exists and equals $\lambda (A)\lambda (B)(A\circ B)^\wedge$.

Conversely, suppose $\escript$ is strongly comparable and inverse-preserving. We have previously observed that $\escript$ is automatically spectral. Applying Theorem~\ref{thm52} there exists a COSEA isomorphism $J$ from $\escript$ onto a Hilbertian sub-COSEA $\fscript$ of
$\escript (H)$. Define the product $J(a)\ctimes J(b)=J(a\circ b)$ on $\fscript$. It is shown in \cite{gud18} that $\fscript$ becomes a COSEA under this product. If $a\in\escript$ is invertible, then $J(a)$ is invertible with $J(a)^{-1}=J(a^{-1})$. Indeed, $\doubleab{J(a^{-1})}=1$, $J(a^{-1})\mid J(a)$ and
\begin{equation*}
J(a)\ctimes J(a^{-1})=J(a\circ a^{-1})=J(\lambda (a)1)=\lambda (a)I
\end{equation*}
If $\escript$ is inverse preserving, then $\ctimes$ is also inverse preserving because if $J(a)$ and $J(b)$ are invertible, then $a$ and $b$ are invertible and
\begin{align*}
\sqbrac{J(a)\ctimes J(b)}^{-1}&=\sqbrac{J(a\circ b)}^{-1}=J\sqbrac{(a\circ b)^{-1}}=J(a^{-1}\circ b^{-1})\\
&=J(a^{-1})\ctimes J(b^{-1})=J(a)^{-1}\ctimes J(b)^{-1}
\end{align*}
We conclude that $\fscript\subseteq\escript (H)$ is an inverse preserving COSEA with sequence product. But $\fscript$ is also an inverse preserving COSEA under the standard sequential product $\circ$. It follows from Theorem~5.19 in \cite{wet181} that
$J(a)\ctimes J(b)=J(a)\circ J(b)$. Hence, $J\colon\escript\to\fscript$ is a COSEA isomorphism.
\end{proof}

\section{Closing Comments} 
A natural question the reader may ask is: ``What is the relationship between contexts as discussed here and the concept of contextuality considered in the literature \cite{ab11,lwe18,spe05}?'' We shall devote a few sentences to this question and leave a more complete investigation to a future work. The notion of contextuality is based on an ontological model for a quantum system. Such a model is described by a measurable space
$(\Lambda,\Sigma )$ where $\Lambda$ is the set of pure states for the system. Preparation procedures, state transformations and measurements are defined by stochastic maps on $\Lambda$ that satisfy certain conditions. One of the main assumptions is that these maps combine to reproduce the experimental statistics of the system in terms of conditional probabilities. We define preparation, transformation and measurement non-contextuality when these stochastic maps satisfy injectiveness properties. Our point is that the concept of contexts can be employed to construct such ontological models by defining the stochastic maps on contexts. Conversely, the stochastic maps for an ontological model will have their supports precisely on the contexts that we have defined in this paper.

Finally, we should mention that other approaches to the mathematical foundations of quantum mechanics have been recently explored. In particular, there have been recent efforts to provide a new foundation for the Hilbert space framework of quantum theory \cite{cdp11,coe10,har12}. The main difference is that these works emphasize the role of composite systems and general transformations, while the COSEA formalism focuses on individual systems and on transformations induced by conditioning with sharp effects.


\begin{thebibliography}{99}
\providecommand{\urlalt}[2]{\href{#1}{#2}}
\providecommand{\doi}[1]{doi:\urlalt{http://dx.doi.org/#1}{#1}}
\providecommand{\arxiv}[1]{\urlalt{https://arxiv.org/abs/#1}{#1}}
\bibitem{ab11}A.~Abramsky and A.~Brandenburger, The sheaf-theoretic structure of non-locality and contextuality,
\textit{New Journal of Physics}, \textbf{13}, 113036 (2011), \doi{10.1088/1367-2630/13/11/113036} and arXiv: quant-ph \arxiv{1102.0264}.
\bibitem{bug94}S.~Bugajski, Fundamentals of fuzzy probability theory,
\textit{Int. J. Theor. Phys.}, \textbf{35}, 2229--2244 (1996), \doi{10.1007/BF02302443}.
\bibitem{cdp11}G.~Chiribella, G.~M.~D'Ariano and P.~Perinotti,
Informational derivation of quantum theory, \textit{Phys. Rev. A} \textbf{84} (2011), \doi{10.1103/PhysRevA.84.012311}.
\bibitem{coe10}B.~Coecke,
A universe of processes and some of its guises, in H.~Halvorson (Ed.), \textit{Deep Beauty: Understanding the Quantum World
Through Mathematical Innovation}, Cambridge University Press, 129--186, 2010, \doi{10.1017/CBO9780511976971.004}.
\bibitem{dp94}A.~Dvure\v censkij and S.~Pulmannov\'a, Difference posets, effects and quantum measurements,
\textit{Int. J. Theor. Phys.}, \textbf{33}, 819--850 (1994), \doi{10.1007/BF00672820}.
\bibitem{fb94}D.~Foulis and M.~K.~Bennett,
Effect algebras and unsharp quantum logics,
\textit{Found. Phys.} \textbf{24} 1331--1352, (1994), \doi{10.1007/BF02283036}.
\bibitem{gg04}A.~Gheondea and S.~Gudder, Sequential product of quantum effects, 
\textit{Proc. Am. Math. Soc.} \textbf{132}, 503--512, (2004), \doi{10.1090/S0002-9939-03-07063-1}.
%
\bibitem{gud981}S.~Gudder,
Fuzzy probability theory,
\textit{Demonstratio Math.} \textbf{31}, 235-254 (1998), \doi{10.1515/dema-1998-0128}.
\bibitem{gud982}S.~Gudder,
Sharp and unsharp quantum effects,
\textit{Adv. Appl. Math.} \textbf{20}, 169--187 (1998), \doi{10.1006/aama.1997.0575}.
\bibitem{gud99}S.~Gudder,
Convex structures and effect algebras,
\textit{Int. J. Theor. Phys.} \textbf{38}, 3179--3187 (1999), \doi{10.1023/A:1026678114856}.
%
\bibitem{gud18}S.~Gudder, Convex and sequential effect algebras,
arXiv: \arxiv{1802.01265} (2018).
\bibitem{gg02}S.~Gudder and R.~Greechie,
Sequential products on effect algebras,
\textit{Rep. Math. Phys.} \textbf{49}, 87--111 (2002), \doi{10.1016/S0034-4877(02)80007-6}.
\bibitem{gl08}S.~Gudder and F.~Latr\'emoli\`ere,
Characterization of the sequential product on quantum effects,
\textit{J. Math. Phys.} \textbf{49}, 052106 (2008), \doi{10.1063/1.2904475}.
\bibitem{gn01}S.~Gudder and G.~Nagy,
Sequential quantum measurements,
\textit{J. Math. Phys.} \textbf{42}, 5212--5222 (2001), \doi{10.1063/1.1407837}.
%
\bibitem{gp98}S.~Gudder and S.~Pulmannov\'a, 
Representation theorem for convex effect algebras,
\textit{Comment. Math. Univ. Carolinae} \textbf{39.4}, 645--659 (1998).
%
\bibitem{har12}L.~Hardy, Quantum theory from five reasonable axioms, arXiv: \arxiv{quant-ph/0101012} (2012).
%
\bibitem{jp18}A.~Jen\v co\' va and M.~Pl\' avala, On the properties of spectral effect algebras, arXiv: quant-ph \arxiv{1811.12407} v1 (2018).
\bibitem{kra83}K.~Kraus, \textit{States, Effects and Operations}, Springer-Verlag, Berlin, 1983, \doi{10.1007/3-540-12732-1}.
%
\bibitem{lwe18}P.~Lillystone, J.~Wallman and J.~Emerson, Contextuality and the single-qubit stabilizer subtheory, arXiv: \arxiv{1802.06121} v1 (2018).
\bibitem{spe05}R.~Spekkens, Contextuality for preparations, transformations, and unsharp measurements, \textit{Phys.~Rev.~A} \textbf{71},
052108 (2005), \doi{10.1103/PhysRevA.71.052108} and arXiv: \arxiv{quant-ph/0406166} (2004).
\bibitem{wet181}J.~Van~de~Wetering, Three characterizations of the sequential product, \textit{J.~Math.~Phys.} \textbf{59}, 082202 (2018), \doi{10.1063/1.5031089} and arXiv: \arxiv{1803.08453} v1 (2018).
%
\bibitem{wet182}J.~Van~de~Wetering, Sequential measurement characterises quantum theory, arXiv: \arxiv{1803.11139} v1 (2018).
\end{thebibliography}
\end{document}